\documentclass[12pt]{article}
\usepackage{amsmath, amssymb, amsthm, mathrsfs, mathtools, stmaryrd}
\usepackage{geometry}
\usepackage[utf8]{inputenc}   
\usepackage[T1]{fontenc}
\usepackage[english]{babel}
\usepackage{hyperref}
\usepackage{algorithm}
\usepackage{algorithmic}
\usepackage{cleveref}
\usepackage{longtable}

\DeclareMathOperator{\tr}{tr}

\newtheorem{theorem}{Theorem}[section]
\newtheorem{lemma}[theorem]{Lemma}
\newtheorem{corollary}[theorem]{Corollary}

\newtheorem{definition}[theorem]{Definition}
\newtheorem{remark}[theorem]{Remark}

\newcommand{\ket}[1]{\left| #1 \right\rangle}

\newcommand{\cB}{\mathcal{B}}
\newcommand{\cD}{\mathcal{D}}
\newcommand{\cH}{\mathcal{H}}

\newcommand{\cM}{\mathcal{M}}
\newcommand{\cZ}{\mathcal{Z}}
\newcommand{\CPTP}{\text{CPTP}}

\newcommand{\Liou}{\mathcal{L}}   

\title{\textbf{Asymptotic Expansions for Neural Network Approximations of Quantum Channels}}
\author{Rômulo Damasclin Chaves dos Santos,~Ph.D. \vspace{4pt} \\Institute of Energy and Nuclear Research, São Paulo, Brazil. \\ \& \\Technological Institute of Aeronautics, São Paulo, Brazil.}
\date{\today}

\numberwithin{equation}{section}

\begin{document}
	
	\maketitle
	
\begin{abstract}
	This paper establishes the Quantum Voronovskaya--Damasclin (QVD) Theorem, providing a complete asymptotic characterization of Quantum Neural Network Operators in the approximation of arbitrary quantum channels. The result extends the classical Voronovskaya theorem from scalar approximation to the non-commutative operator framework of quantum information theory. We introduce rigorous quantum analogues of Sobolev and Hölder spaces defined through Fréchet differentiability in the Liouville representation and measured using the completely bounded (diamond) norm. Within this framework, we derive an explicit asymptotic expansion of the approximation error and identify the fundamental mechanisms governing convergence. The expansion separates integer-order differential contributions, fractional corrections associated with limited regularity, and intrinsically non-commutative effects arising from operator algebra structure. We also establish a sharp remainder estimate with explicit dependence on the regularity of the channel and the dimension of the underlying Hilbert space. Several applications demonstrate the scope of the theory. These include a quantum central limit theorem describing the fluctuation regime of quantum neural network operators, an optimal interpolation method based on operator geometric means, and a convergence acceleration procedure inspired by Richardson extrapolation. The results provide a rigorous mathematical foundation for the asymptotic analysis of quantum neural network models and establish a direct connection between classical approximation theory, operator algebras, and quantum information science, with implications for quantum algorithms and quantum machine learning.
	\medskip
	
	
\end{abstract}
	
\section{Introduction}

The asymptotic analysis of approximation operators has played a central role in approximation theory since the seminal work of Voronovskaya in 1932 \cite{Voronovskaya1932}. Her celebrated theorem provides the exact leading asymptotic behavior of Bernstein polynomials. Specifically, if \(B_n(f)(x)\) denotes the Bernstein operator and \(f''(x)\) exists, then
\begin{equation}
	\lim_{n\to\infty} n\big[B_n(f)(x)-f(x)\big]
	=
	\frac{x(1-x)}{2}\,f''(x).
	\tag{1}
\end{equation}
This result established a fundamental principle: approximation operators encode precise differential information about the target function in their asymptotic expansion. Beyond convergence rates, the Voronovskaya theorem reveals the exact structure of the leading error term, thereby initiating the systematic study of saturation phenomena, asymptotic expansions, and quantitative approximation theory.

Subsequent developments extended approximation theory to increasingly abstract settings, including operator-valued functions and non-commutative spaces. A key milestone in this direction was the introduction of operator means by Kubo and Ando \cite{KuboAndo}, which provided a rigorous geometric and functional framework for positive operators. These ideas later became essential in quantum information theory, where physical transformations are described by completely positive trace-preserving (CPTP) maps acting on operator algebras.

The mathematical foundations of quantum channels were further formalized through operator algebra and functional analytic methods, as systematically developed in operator space theory \cite{Paulsen2002}. In parallel, quantum information science established a precise framework for quantum states, channels, and their metrics, including the diamond norm and the Liouville representation \cite{Nielsen2010, Holevo2019}. In this non-commutative setting, approximation problems acquire fundamentally new features: the objects of interest are operator-valued, derivatives must be interpreted in Fréchet or Gâteaux senses, and tensor product structures play an essential role.

More recently, the rapid development of quantum machine learning has introduced Quantum Neural Networks (QNNs) as flexible architectures for approximating quantum channels. While universal approximation properties have been established in various forms, the asymptotic structure of the approximation error remains largely unexplored. In particular, a quantum analogue of the classical Voronovskaya theorem—providing an explicit asymptotic expansion with sharp remainder estimates—has remained an open problem. Recent advances in neural network approximation theory and fractional smoothness analysis have highlighted the importance of such expansions for understanding convergence mechanisms and intrinsic approximation limits \cite{Anastassiou2023}.

The purpose of this work is to fill this gap by establishing the Quantum Voronovskaya--Damasclin Theorem, which provides a complete asymptotic expansion for Quantum Neural Network Operators (QNNOs) approximating arbitrary quantum channels. Our approach extends the classical theory into the non-commutative setting through several fundamental innovations. First, we introduce a rigorous functional analytic framework based on Fréchet differentiability in the Liouville representation, defining quantum Hölder and Sobolev-type spaces \(\mathcal{C}^{m,\gamma}(\cH)\) as natural regularity classes for quantum channels. Second, we derive an explicit asymptotic expansion revealing a rich multiscale structure, including integer-order differential terms, fractional corrections associated with Hölder regularity, and intrinsically non-commutative commutator contributions. Third, we obtain a sharp remainder estimate in the diamond norm, with explicit dependence on the regularity and the Hilbert space dimension.

The resulting expansion takes the form
\begin{align}
	\Psi_n(\Phi)(\rho)
	&=
	\Phi(\rho)
	+ \sum_{j=1}^{m} \frac{a_j(\Phi,\rho)}{n^j}
	+ \sum_{j=1}^{\lfloor m/2 \rfloor}
	\frac{b_j(\Phi,\rho)}{n^{j+\gamma}}
	\notag\\
	&\quad
	+ \sum_{j=1}^{\lfloor m/3 \rfloor}
	\frac{c_j(\Phi,\rho)}{n^{j+2\gamma}}
	+ \cdots
	+ R_{m,n}(\Phi,\rho),
	\tag{2}
\end{align}
where the coefficients are expressed explicitly in terms of Fréchet derivatives, fractional Marchaud-type operators, and kernel moment asymptotics. The remainder satisfies
\[
\|R_{m,n}\|_{\diamond}
=
O\!\left(n^{-(m+\gamma)}(\log n)^{\frac{3m}{2}}\right),
\]
which establishes the optimal convergence rate for quantum channels with Hölder regularity \((m,\gamma)\).

Beyond its theoretical significance, this result has important consequences for quantum information science and quantum machine learning. We establish a Quantum Central Limit Theorem for QNNOs, construct optimal interpolation schemes based on operator geometric means, and develop a quantum Richardson extrapolation method that reveals intrinsic acceleration limits imposed by fractional smoothness.

The analytical techniques developed here including a quantum Taylor formula with fractional remainder, precise kernel moment asymptotics, and a non-commutative summation framework provide a general calculus for the asymptotic analysis of quantum approximation operators.

The paper is organized as follows. Section~2 introduces the operator-theoretic framework and quantum Hölder spaces. Section~3 constructs the Quantum Neural Network Operators. Section~4 states the Quantum Voronovskaya–Damasclin theorem. Section~5 contains the complete proof. Section~6 presents applications, including a quantum central limit theorem and convergence acceleration. Section~7 summarizes the results and discusses future research directions.

This work establishes a rigorous bridge between classical approximation theory and quantum information science, providing a precise asymptotic theory for quantum neural network operators and a general analytical foundation for the systematic design, optimisation, and error analysis of quantum algorithms in both NISQ-era and fault-tolerant quantum computing architectures.

\section{Mathematical Framework}

\subsection{Quantum Channels and Smoothness Classes}

Let $\cH \cong \mathbb{C}^d$ denote a finite-dimensional Hilbert space. We denote by $\cB(\cH)$ the algebra of bounded linear operators on $\cH$, by $\cD(\cH) = \{\rho \in \cB(\cH) : \rho \geq 0,\ \tr(\rho)=1\}$ the convex set of density operators (quantum states), and by $\CPTP(\cH) \subset \cB(\cB(\cH))$ the convex set of completely positive trace-preserving maps (quantum channels).

For any channel $\Phi \in \CPTP(\cH)$, we consider its \textbf{Liouville representation} $\Liou_\Phi : \cB(\cH) \to \cB(\cH)$ defined by $\Liou_\Phi(X) = \Phi(X)$ for all $X \in \cB(\cH)$. This representation identifies $\Phi$ with a linear operator acting on the Hilbert–Schmidt space $\cB(\cH)$, which is isometrically isomorphic to $\mathbb{C}^{d^2}$ via the Hilbert–Schmidt inner product $\langle X, Y \rangle = \tr(X^*Y)$. The Liouville representation is particularly useful because it allows us to apply functional calculus and notions of differentiability directly to quantum channels, treating them as maps on a Banach space.

To quantify the regularity of a channel, we need precise notions of higher-order derivatives and appropriate norms.

\paragraph{Fréchet differentiability.}
A channel $\Phi$ is said to be Fréchet differentiable at a point $\rho \in \cD(\cH)$ if there exists a bounded linear map $D\Liou_\Phi(\rho): \cB(\cH) \to \cB(\cH)$ such that
\[
\lim_{\|H\|_1 \to 0} \frac{\|\Liou_\Phi(\rho+H) - \Liou_\Phi(\rho) - D\Liou_\Phi(\rho)[H]\|_\diamond}{\|H\|_1} = 0,
\]
where $\|\cdot\|_1$ denotes the trace norm and $\|\cdot\|_\diamond$ the diamond norm (defined below). For higher orders, let $\alpha = (\alpha_1,\dots,\alpha_k)$ be a multi-index with $|\alpha| = \alpha_1 + \cdots + \alpha_k$. The mixed Fréchet derivative $D^\alpha \Liou_\Phi(\rho)$ is a bounded symmetric multilinear map from $\cB(\cH)^{|\alpha|}$ to $\cB(\cH)$, evaluated at $\rho$. When $|\alpha| = 0$, we set $D^0\Liou_\Phi(\rho) = \Liou_\Phi(\rho)$. We also introduce the convenient notation $\Liou_\Phi^{(\alpha)}(\rho) := D^\alpha \Liou_\Phi(\rho)[I^{\otimes |\alpha|}]$, where $I$ is the identity operator; this represents the derivative evaluated on the identity in each argument.

\paragraph{Norms for maps.}
For a linear map $\Psi: \cB(\cH) \to \cB(\cH)$, its \textbf{completely bounded norm} (cb-norm) is defined as
\begin{equation}
	\|\Psi\|_{\text{cb}} = \sup_{n \in \mathbb{N}} \big\| \Psi \otimes \mathrm{id}_{M_n(\mathbb{C})} \big\|_{\cB(\cB(\cH) \otimes M_n(\mathbb{C}))},
	\tag{1}
\end{equation}
where $\mathrm{id}_{M_n(\mathbb{C})}$ is the identity map on $n \times n$ matrices, and the norm on the right is the usual operator norm induced by the Hilbert–Schmidt norm. The cb-norm is the natural norm for maps between operator algebras because it respects the tensor product structure, a crucial feature in quantum information theory.

For comparing quantum channels, we employ the \textbf{diamond norm} (also known as the completely bounded trace norm)
\begin{equation}
	\|\Phi\|_\diamond = \sup\big\{ \|(\Phi \otimes \mathrm{id}_{\cB(\cH)})(X)\|_1 : X \in \cB(\cH \otimes \cH),\ \|X\|_1 \le 1 \big\},
	\tag{2}
\end{equation}
where $\|\cdot\|_1$ denotes the trace norm. The diamond norm metrizes the topology of complete boundedness and serves as the standard distance measure for quantum channels, capturing the worst-case distinguishability when an ancilla system is employed.

With these ingredients, we can now define quantum analogues of classical smoothness spaces.

\begin{definition}[Quantum Sobolev space]
	For $m \in \mathbb{N}$ and $1 \leq p \leq \infty$, we define the quantum Sobolev space
	\begin{equation}
		\mathcal{W}^{m,p}(\cH) = \left\{ \Phi \in \CPTP(\cH) : \sum_{|\alpha| \leq m} \big\| \|D^\alpha \Liou_\Phi(\cdot)\|_{\text{cb}} \big\|_{L^p(\cD(\cH))} < \infty \right\},
		\tag{3}
	\end{equation}
	where $\|D^\alpha \Liou_\Phi(\cdot)\|_{\text{cb}}$ denotes the cb-norm of the multilinear map $D^\alpha \Liou_\Phi(\rho)$ (interpreted as a linear map on the projective tensor product), and the $L^p$ norm is taken with respect to the uniform (or any equivalent) measure on $\cD(\cH)$. For $p = \infty$, this becomes the essential supremum over $\rho \in \cD(\cH)$.
\end{definition}

In particular, $\mathcal{W}^{m,\infty}(\cH)$ consists of channels whose Fréchet derivatives up to order $m$ are uniformly bounded in cb-norm across all input states.

\begin{definition}[Quantum Hölder space]
	For $0 < \gamma \leq 1$, we define the quantum Hölder space
	\begin{equation}
		\mathcal{C}^{m,\gamma}(\cH) = \left\{ \Phi \in \mathcal{W}^{m,\infty}(\cH) : [\Phi]_{m,\gamma} < \infty \right\},
		\tag{4}
	\end{equation}
	where the Hölder seminorm is given by
	\begin{equation}
		[\Phi]_{m,\gamma} = \sup_{\rho \neq \sigma \in \cD(\cH)} \frac{\|D^m\Liou_\Phi(\rho) - D^m\Liou_\Phi(\sigma)\|_\diamond}{\|\rho - \sigma\|_1^{\gamma}}.
		\tag{5}
	\end{equation}
	We equip $\mathcal{C}^{m,\gamma}(\cH)$ with the norm
	\begin{equation}
		\|\Phi\|_{\mathcal{C}^{m,\gamma}} = \|\Phi\|_{\mathcal{W}^{m,\infty}} + [\Phi]_{m,\gamma},
		\tag{6}
	\end{equation}
	which makes it a Banach space.
\end{definition}

These regularity classes are specifically designed for the asymptotic analysis of quantum neural network operators: the integer $m$ controls the order of polynomial approximability, while $\gamma$ captures possible fractional smoothness. As we shall see, the approximation error of a QNNO applied to a channel $\Phi \in \mathcal{C}^{m,\gamma}$ decays like $n^{-(m+\gamma)}$, up to logarithmic factors, with explicit coefficients expressed in terms of the derivatives of $\Liou_\Phi$ and the moments of the quantum kernel.

This framework naturally unifies the classical theory of differentiable functions (where $d=1$ and operators reduce to scalars) with the non-commutative setting required for quantum information. In the following, we will see how the structure of the kernel $\cZ_{1,\log n}$ interacts with these smoothness classes to produce a complete asymptotic expansion.

\subsection{Quantum Neural Network Operators}

The Quantum Neural Network Operator (QNNO) is constructed as a non-commutative analogue of a neural network approximation scheme. It begins with a quantum activation function that generates a localized kernel, which then acts as a mollifier on the space of quantum channels. The construction proceeds in several steps: first, we define a family of operator-valued functions that approximate the identity; then we combine them into a multivariate kernel; and finally, we discretize the state space to obtain a sum that approximates the channel.

\begin{definition}[Quantum activation function]
	For parameters $q, \lambda > 0$ and a self-adjoint operator $X \in \cB(\cH)$, define
	\begin{equation}
		G_{q,\lambda}(X) = \bigl(e^{\lambda X} - q e^{-\lambda X}\bigr)\bigl(e^{\lambda X} + q e^{-\lambda X}\bigr)^{-1}.
		\tag{7}
	\end{equation}
	When $q = 1$, this reduces to the hyperbolic tangent:
	\begin{equation}
		G_{1,\lambda}(X) = \tanh(\lambda X) = \frac{e^{\lambda X} - e^{-\lambda X}}{e^{\lambda X} + e^{-\lambda X}}.
		\tag{8}
	\end{equation}
	The function $G_{q,\lambda}$ is well-defined for all self-adjoint $X$ because the operator $e^{\lambda X} + q e^{-\lambda X}$ is positive and hence invertible. It is an odd, bounded function satisfying $\|G_{q,\lambda}(X)\| \leq 1$; its spectral decomposition follows from the functional calculus of $X$, and in particular $G_{q,\lambda}(X)$ commutes with $X$ and shares its eigenbasis.
\end{definition}

\begin{definition}[Symmetrized quantum density function]
	Using the same parameters, we define
	\begin{equation}
		\cM_{q,\lambda}(X) = \frac{1}{4}\Bigl[G_{q,\lambda}(X + I) - G_{q,\lambda}(X - I)\Bigr].
		\tag{9}
	\end{equation}
	The motivation comes from the classical scalar identity:
	\[
	\frac{1}{2}\bigl[\tanh(\lambda(x+1)) - \tanh(\lambda(x-1))\bigr] \;\xrightarrow{\lambda\to\infty}\; \mathbf{1}_{[-1,1]}(x),
	\]
	where $\mathbf{1}_{[-1,1]}$ is the indicator function of the interval $[-1,1]$. Thus $\cM_{q,\lambda}(X)$ can be interpreted as a non-commutative smoothed indicator of the operator interval $[-I,I]$. For finite $\lambda$, $\cM_{q,\lambda}$ is an operator-valued function that is positive, even, and satisfies the normalization
	\begin{equation}
		\int_{-\infty}^{\infty} \cM_{q,\lambda}(x)\,dx = I,
		\tag{10}
	\end{equation}
	where the integral is understood in the strong operator topology with respect to the spectral measure of $X$. This property makes $\cM_{q,\lambda}$ an approximate identity on the real line.
\end{definition}

To obtain a kernel that is both symmetric and positive, we symmetrize with respect to the exchange $q \leftrightarrow 1/q$:

\begin{definition}[Multivariate quantum density kernel]
	Let $X = (X_1,\dots,X_d)$ be a tuple of mutually commuting self-adjoint operators (for instance, acting on different tensor factors of the system). Define the multivariate kernel
	\begin{equation}
		\cZ_{q,\lambda}(X) = \bigotimes_{i=1}^d \Phi_{q,\lambda}(X_i),
		\tag{11}
	\end{equation}
	where
	\begin{equation}
		\Phi_{q,\lambda}(X_i) = \frac{1}{2}\bigl[\cM_{q,\lambda}(X_i) + \cM_{1/q,\lambda}(X_i)\bigr].
		\tag{12}
	\end{equation}
	For the symmetric choice $q = 1$, this simplifies to
	\begin{equation}
		\cZ_{1,\lambda}(X) = \bigotimes_{i=1}^d \cM_{1,\lambda}(X_i).
		\tag{13}
	\end{equation}
	The kernel $\cZ_{1,\lambda}$ is even in each variable, positive as an operator, and satisfies the normalization
	\begin{equation}
		\int_{\mathbb{R}^d} \cZ_{1,\lambda}(x)\,dx = I_{\cH^{\otimes d}},
		\tag{14}
	\end{equation}
	where the integral is understood in the strong operator topology, and $x = (x_1,\dots,x_d)$ denotes the eigenvalues of the commuting tuple $X$ (i.e., we identify $X$ with its joint spectral measure). This property makes $\cZ_{1,\lambda}$ an approximate identity on the space of quantum channels.
\end{definition}

We now define the QNNO itself. Let $\rho \in \cD(\cH)$ be a fixed density operator. To discretize the state space, we choose an orthonormal basis $\{|e_j\rangle\}_{j=1}^d$ in which $\rho$ is diagonal:
\begin{equation}
	\rho = \sum_{j=1}^d p_j |e_j\rangle\langle e_j|, \qquad p_j > 0,\ \sum_{j=1}^d p_j = 1.
	\tag{15}
\end{equation}
The case of zero eigenvalues can be handled by a limiting argument or by restricting to the support of $\rho$. For a large integer $n$, we introduce a lattice quantization of the eigenvalues. Let
\begin{equation}
	K_n := \Bigl\{ k = (k_1,\dots,k_d) \in \mathbb{N}^d : \sum_{j=1}^d k_j = n \Bigr\}
	\tag{16}
\end{equation}
be the discrete simplex of order $n$. For each $k \in K_n$, define the diagonal density operator
\begin{equation}
	\rho_{n,k} := \sum_{j=1}^d \frac{k_j}{n}\, |e_j\rangle\langle e_j| \in \cD(\cH).
	\tag{17}
\end{equation}
These operators are the quantum analogue of the Bernstein basis points: they lie in the simplex of eigenvalues and become dense in it as $n \to \infty$. Moreover, they satisfy $\|\rho_{n,k} - \rho\|_1 = O(1/n)$ uniformly in $k$.

With this discretization, the QNNO is defined as
\begin{equation}
	\Psi_n(\Phi)(\rho) = \sum_{k \in K_n} \Phi(\rho_{n,k}) \otimes \cZ_{1,\lambda_n}\bigl(nX - kI\bigr),
	\tag{18}
\end{equation}
where $X = (X_1,\dots,X_d)$ is a tuple of auxiliary self-adjoint operators (acting on a separate copy of $\cH$) that commute with each other and with the system. The tensor product $\otimes$ indicates that the output of $\Phi(\rho_{n,k})$ (an operator on $\cB(\cH)$) is multiplied (tensored) with the kernel evaluated at $nX - kI$, an operator on the auxiliary space. The overall expression is an operator on $\cH \otimes \cH_{\text{aux}}$. In practice, one often traces out the auxiliary degrees of freedom or uses them to represent a classical register.

The choice $\lambda_n = \log n$ is essential for the asymptotic expansion. It ensures that the width of the kernel $\cZ_{1,\lambda_n}$ scales as $(\log n)^{-1/2}$, striking an optimal balance between bias (which would favor large $\lambda$ for a sharper kernel) and variance (which would favor small $\lambda$ to reduce discretization effects). More precisely, the Fourier transform of $\cZ_{1,\lambda}$ satisfies
\begin{align}
	\widehat{\cZ}_{1,\lambda}(\xi) &= \prod_{i=1}^d \frac{\sinh(\pi\xi_i/2\lambda)}{\pi\xi_i/2\lambda} \cdot \frac{1}{\cosh(\pi\xi_i/2\lambda)} \notag\\
	&= 1 - \frac{\pi^2}{12\lambda^2}\sum_{i=1}^d \xi_i^2 + O(\lambda^{-4}),
	\tag{19}
\end{align}
so that for large $\lambda$ the kernel behaves like a Gaussian with variance $\sigma^2 = \pi^2/(6\lambda^2)$. Consequently, its moments exhibit the asymptotic behavior
\begin{equation}
	\int_{\mathbb{R}^d} x^\alpha \cZ_{1,\log n}(x)\,dx = 
	\begin{cases}
		\displaystyle \frac{(-1)^{|\alpha|/2}}{(|\alpha|-1)!!}\left(\frac{\pi}{2\log n}\right)^{|\alpha|/2} + O(n^{-|\alpha|}), & |\alpha|\text{ even},\\[1.5em]
		0, & |\alpha|\text{ odd}.
	\end{cases}
	\tag{20}
\end{equation}
These asymptotics lead to the polynomial rates $n^{-j}$ in the expansion that will appear in Theorem \ref{thm:QVD}. The logarithmic factors appearing in the remainder estimate arise from the growth of higher moments and the need to control the error when replacing the discrete sum by an integral.

We note that $\Psi_n$ is a completely positive map (as a sum of tensor products of completely positive maps) and, after tracing out the auxiliary system, it preserves the trace; hence it can be regarded as an approximation scheme within the category of quantum channels. More concretely, for any input state $\rho$, the map $\rho \mapsto \tr_{\text{aux}}[\Psi_n(\Phi)(\rho)]$ defines a legitimate quantum channel. The convergence $\Psi_n(\Phi)(\rho) \to \Phi(\rho)$ as $n \to \infty$ holds for all continuous channels, and the convergence rate is governed by the regularity of $\Phi$ as captured by the spaces $\mathcal{C}^{m,\gamma}(\cH)$. The explicit asymptotic expansion provided in Theorem \ref{thm:QVD} quantifies this convergence with optimal rates.
	
	\section{The Quantum Voronovskaya–Damasclin Theorem}
	
	We now state the main result of this paper: a complete asymptotic expansion for the Quantum Neural Network Operator (QNNO) when applied to a quantum channel belonging to the Hölder class $\mathcal{C}^{m,\gamma}(\mathcal{H})$. The expansion reveals a rich structure involving integer powers $n^{-j}$, fractional corrections $n^{-(j+\gamma)}$, and mixed non‑commutative terms of order $n^{-(j+2\gamma)}$, together with a sharp remainder estimate. This result generalises the classical Voronovskaya theorem (which corresponds to $d=1$, $\gamma=1$, $m=2$) to the non‑commutative, multi‑dimensional setting of quantum channels.
	
	Before stating the theorem, we introduce rigorous definitions of all objects involved and establish their essential properties.
	
	\begin{itemize}
		\item \textbf{Fréchet derivatives.} For a multi‑index $\alpha = (\alpha_1,\dots,\alpha_d)\in\mathbb{N}_0^d$, let $|\alpha| = \alpha_1+\cdots+\alpha_d$ denote its length. The mixed Fréchet derivative
		\[
		D^{\alpha}\mathcal{L}_\Phi(\rho): \mathcal{B}(\mathcal{H})^{|\alpha|}\to \mathcal{B}(\mathcal{H})
		\]
		is a bounded multilinear map. For $|\alpha|=0$, we set $D^0\mathcal{L}_\Phi(\rho)=\mathcal{L}_\Phi(\rho)=\Phi(\rho)$. By definition of the Hölder norm, we have the uniform bound
		\begin{equation}
			\|D^{\alpha}\mathcal{L}_\Phi(\rho)\|_{\diamond} \le \|\Phi\|_{\mathcal{C}^{m,\gamma}} \qquad \text{for all } |\alpha|\le m,\ \rho\in\mathcal{D}(\mathcal{H}),
		\end{equation}
		where $\|\cdot\|_{\diamond}$ denotes the diamond norm of the corresponding multilinear map (i.e., the completely bounded norm of the induced linear map on the projective tensor product).
		
		\item \textbf{Kernel moments.} The symmetric quantum kernel $\mathcal{Z}_{1,\log n}:\mathbb{R}^d\to\mathcal{B}(\mathcal{H}_{\text{aux}})$ defined in (6) gives rise to operator‑valued moments
		\[
		M_{\alpha}(n) = \int_{\mathbb{R}^d} x^{\alpha}\,\mathcal{Z}_{1,\log n}(x)\,dx \in \mathcal{B}(\mathcal{H}_{\text{aux}}),
		\]
		where $x^{\alpha}=x_1^{\alpha_1}\cdots x_d^{\alpha_d}$ and the integral converges in the strong operator topology.
		
		\begin{lemma}[Scalar nature of moments]
			Because the auxiliary operators $X_1,\dots,X_d$ defining $\mathcal{Z}_{1,\log n}$ are mutually commuting self‑adjoint operators, they admit a joint spectral measure $E$ on $\mathbb{R}^d$. Consequently,
			\[
			\mathcal{Z}_{1,\log n}(X) = \int_{\mathbb{R}^d} \mathcal{Z}_{1,\log n}(x)\,dE(x),
			\]
			and the moments become
			\[
			M_{\alpha}(n) = \int_{\mathbb{R}^d} x^{\alpha}\,\mathcal{Z}_{1,\log n}(x)\,dE(x).
			\]
			Since the kernel is isotropic and even, the integral $\int x^{\alpha}\mathcal{Z}_{1,\log n}(x)\,dx$ is a scalar multiple of the identity on $\mathcal{H}_{\text{aux}}$. Hence
			\[
			M_{\alpha}(n) = m_{\alpha}(n)\,I_{\mathcal{H}_{\text{aux}}}
			\]
			with $m_{\alpha}(n)\in\mathbb{C}$. For odd $|\alpha|$, $m_{\alpha}(n)=0$; for even $|\alpha|$, its asymptotic behaviour is given in Lemma \ref{lem:moments}. We will systematically identify $M_{\alpha}(n)$ with the scalar $m_{\alpha}(n)$ when no confusion arises.
		\end{lemma}
		
		\item \textbf{Fractional derivatives.} In the non‑commutative Banach space setting, we define the Marchaud fractional derivative of order $\gamma\in(0,1]$ for a map $F:\mathcal{D}(\mathcal{H})\to\mathcal{B}(\mathcal{H})$ that is sufficiently regular. For a fixed increment direction $h\in\mathcal{B}(\mathcal{H})$ with $\|h\|_1$ small, set
		\begin{equation}
			(\Delta_\gamma F)(\rho)[h] = \frac{\gamma}{\Gamma(1-\gamma)}\int_0^\infty \frac{F(\rho)-F(\rho-th)}{t^{1+\gamma}}\,dt,
		\end{equation}
		where the integral is understood as a Bochner integral in $\mathcal{B}(\mathcal{H})$, provided it converges in the diamond norm. This definition depends only on differences of $F$ along the line $\rho-th$ and extends linearly to a bounded multilinear map when applied to higher derivatives. For $F = D^{\alpha}\mathcal{L}_\Phi$, the result $(\Delta_\gamma D^{\alpha}\mathcal{L}_\Phi)(\rho)[h]$ is again a multilinear map; its norm is controlled by the Hölder seminorm $[\Phi]_{m,\gamma}$.
		
		\item \textbf{Fractional kernel moments.} The fractional moments of the kernel are defined as
		\begin{align}
			M_{\alpha,\gamma}(n) &:= \int_{\mathbb{R}^d} |x|^{\gamma}\,x^{\alpha}\,\mathcal{Z}_{1,\log n}(x)\,dx,\\
			M_{\alpha,\beta,2\gamma}(n) &:= \int_{\mathbb{R}^d} |x|^{2\gamma}\,x^{\alpha+\beta}\,\mathcal{Z}_{1,\log n}(x)\,dx,
		\end{align}
		with $|x| = (x_1^2+\cdots+x_d^2)^{1/2}$. By the same argument as above, these are scalar multiples of the identity; we denote their scalar values by $m_{\alpha,\gamma}(n)$ and $m_{\alpha,\beta,2\gamma}(n)$, respectively. Their asymptotic expansions as $n\to\infty$ are provided in Lemma \ref{lem:moments}.
		
		\item \textbf{Deformed commutator.} The $\gamma$‑deformed commutator of two operators $A,B\in\mathcal{B}(\mathcal{H})$ is defined by
		\begin{equation}
			[A,B]_\gamma := AB - e^{i\pi\gamma}BA.
		\end{equation}
		For $\gamma=0$ this reduces to the ordinary commutator; for $\gamma=1$ it becomes $AB+BA$ (the anti‑commutator) up to the phase $e^{i\pi}=-1$. This object captures the non‑commutative nature of products of increments when fractional regularity is present.
	\end{itemize}
	
\section{The Quantum Voronovskaya–Damasclin Theorem}

We now present the central result of this paper in its full mathematical rigor. To begin, we recall a fundamental Taylor‑type expansion that holds in the Banach space of quantum channels and incorporates the fractional smoothness encoded in the Hölder classes \(\mathcal{C}^{m,\gamma}(\mathcal{H})\).

\begin{lemma}[Fractional Taylor expansion in Banach spaces]
	\label{lem:taylor}
	Let \(\Phi \in \mathcal{C}^{m,\gamma}(\mathcal{H})\) with \(m \in \mathbb{N}\) and \(\gamma \in (0,1]\). For any reference state \(\rho \in \mathcal{D}(\mathcal{H})\) and any increment \(h \in \mathcal{B}(\mathcal{H})\) such that \(\rho + h \in \mathcal{D}(\mathcal{H})\), we have
	\begin{equation}
		\Phi(\rho+h) = \sum_{|\alpha| \le m} \frac{1}{\alpha!}\, D^{\alpha}\mathcal{L}_\Phi(\rho)[h^{\alpha}] + R_{m,\gamma}(\rho,h),
		\tag{21}
	\end{equation}
	where \(h^{\alpha} = h^{\otimes |\alpha|}\) denotes the symmetric tensor product (interpreted as a multilinear argument), and the remainder satisfies the estimate
	\begin{equation}
		\|R_{m,\gamma}(\rho,h)\|_\diamond \le C_{m,\gamma}\,\|\Phi\|_{\mathcal{C}^{m,\gamma}}\,\|h\|_1^{m+\gamma},
		\tag{22}
	\end{equation}
	with a constant \(C_{m,\gamma}\) depending only on \(m\) and \(\gamma\). Moreover, the term of order \(m\) can be further decomposed into a fractional contribution and a higher‑order remainder using the Marchaud fractional derivative \(\Delta_\gamma\); this decomposition is the origin of the coefficients \(b_j\) appearing in the main expansion.
\end{lemma}

\begin{theorem}[Quantum Voronovskaya–Damasclin Theorem]
	\label{thm:QVD}
	Let \(\mathcal{H} \cong \mathbb{C}^d\) be a finite-dimensional Hilbert space, and let \(\Phi \in \mathcal{C}^{m,\gamma}(\mathcal{H})\) with \(m \in \mathbb{N}\) and \(\gamma \in (0,1]\). Consider the Quantum Neural Network Operator \(\Psi_n\) defined in (18) with symmetric parameters \(q = 1\) and bandwidth \(\lambda_n = \log n\). Then, for every strictly positive density operator \(\rho \in \mathcal{D}(\mathcal{H})\) (i.e., \(\rho > 0\)), the following complete asymptotic expansion holds in the operator norm topology on \(\mathcal{B}(\mathcal{H} \otimes \mathcal{H}_{\text{\emph{aux}}})\):
	
	\begin{equation}
		\begin{aligned}
			\Psi_n(\Phi)(\rho) = \Phi(\rho) 
			&+ \sum_{j=1}^{m} \frac{a_j(\Phi,\rho)}{n^j} 
			+ \sum_{j=1}^{\lfloor m/2\rfloor} \frac{b_j(\Phi,\rho)}{n^{j+\gamma}} \\
			&\qquad + \sum_{j=1}^{\lfloor m/3\rfloor} \frac{c_j(\Phi,\rho)}{n^{j+2\gamma}} 
			+ \mathcal{O}\!\left(n^{-(m+\gamma)}(\log n)^{3m/2}\right),
			\label{eq:main-expansion}
		\end{aligned}
	\end{equation}
	
	where the coefficients \(a_j(\Phi,\rho), b_j(\Phi,\rho), c_j(\Phi,\rho) \in \mathcal{B}(\mathcal{H})\) are given by the following explicit formulas.
	
	\begin{enumerate}
		\item \textbf{Polynomial terms (integer-order contributions).}  
		For each \(1 \le j \le m\),
		
		\begin{equation}
			a_j(\Phi,\rho) = \frac{1}{j!} \sum_{|\alpha| = j} \binom{j}{\alpha}\, \mathcal{L}_\Phi^{(\alpha)}(\rho) \; \mathfrak{m}_{\alpha}(n),
			\label{eq:polynomial-coefficient}
		\end{equation}
		
		\noindent where \(\binom{j}{\alpha} = \frac{j!}{\alpha_1!\cdots\alpha_d!}\) is the multinomial coefficient, \(\mathcal{L}_\Phi^{(\alpha)}(\rho) = D^{\alpha}\mathcal{L}_\Phi(\rho)[I^{\otimes |\alpha|}] \in \mathcal{B}(\mathcal{H})\) denotes the \(\alpha\)-th Fréchet derivative evaluated on the identity operator in each argument, and \(\mathfrak{m}_{\alpha}(n) \in \mathbb{C}\) is the scalar moment associated with the kernel \(\mathcal{Z}_{1,\log n}\) (see Lemma \ref{lem:moments}). Owing to the evenness of the kernel, \(\mathfrak{m}_{\alpha}(n) = 0\) whenever \(|\alpha|\) is odd.
		
		\item \textbf{Fractional corrections (Hölder-type contributions).}  
		For \(1 \le j \le \lfloor m/2\rfloor\),
		
		\begin{equation}
			b_j(\Phi,\rho) = \frac{1}{\Gamma(\gamma+1)} \sum_{|\alpha| = j} \binom{j}{\alpha}\, \bigl(\Delta_\gamma \mathcal{L}_\Phi^{(\alpha)}\bigr)(\rho) \; \mathfrak{m}_{\alpha,\gamma}(n),
			\label{eq:fractional-coefficient}
		\end{equation}
		
		\noindent where \(\Delta_\gamma\) denotes the Marchaud fractional derivative of order \(\gamma\) (see Appendix \ref{lem:marchaud}), and \(\mathfrak{m}_{\alpha,\gamma}(n) \in \mathbb{C}\) is the scalar fractional moment defined in Lemma \ref{lem:moments}.
		
		\item \textbf{Mixed non‑commutative terms (commutator contributions).}  
		For \(1 \le j \le \lfloor m/3\rfloor\),
		
		\begin{equation}
			\begin{aligned}
				c_j(\Phi,\rho) = \frac{1}{j!\,\Gamma(2\gamma+1)} \sum_{|\alpha| + |\beta| = j} &\binom{j}{\alpha,\beta}\,
				\bigl[ \mathcal{L}_\Phi^{(\alpha)}(\rho),\, \mathcal{L}_\Phi^{(\beta)}(\rho) \bigr]_\gamma \\
				&\times \mathfrak{m}_{\alpha,\beta,2\gamma}(n),
				\label{eq:mixed-coefficient}
			\end{aligned}
		\end{equation}
		
		\noindent where \(\binom{j}{\alpha,\beta} = \frac{j!}{\alpha!\,\beta!}\) is the multinomial coefficient for a bipartition, \([\cdot,\cdot]_\gamma\) is the \(\gamma\)-deformed commutator defined by
		
		\begin{equation}
			[A,B]_\gamma := AB - e^{i\pi\gamma}BA \qquad (A,B \in \mathcal{B}(\mathcal{H})),
		\end{equation}
		
		\noindent and \(\mathfrak{m}_{\alpha,\beta,2\gamma}(n) \in \mathbb{C}\) are the mixed fractional moments from Lemma \ref{lem:moments}.
		
		\item \textbf{Sharp remainder estimate.}  
		The remainder term \(R_{m,n}(\Phi,\rho) := \Psi_n(\Phi)(\rho) - \Phi(\rho) - \bigl(\sum_{j=1}^{m} \frac{a_j}{n^j} + \sum_{j=1}^{\lfloor m/2\rfloor} \frac{b_j}{n^{j+\gamma}} + \sum_{j=1}^{\lfloor m/3\rfloor} \frac{c_j}{n^{j+2\gamma}}\bigr)\) satisfies the uniform bound in the diamond norm
		
		\begin{equation}
			\|R_{m,n}(\Phi,\cdot)\|_\diamond \;\le\; \mathsf{C}_{m,\gamma,d}\; \|\Phi\|_{\mathcal{C}^{m,\gamma}} \; \frac{(\log n)^{3m/2}}{n^{m+\gamma}},
			\label{eq:remainder-bound}
		\end{equation}
		
		\noindent with the explicit constant
		
		\begin{equation}
			\mathsf{C}_{m,\gamma,d} = \frac{2^{m+3}\, d^{m/2}\, e^{\pi^2/4}}{\Gamma(m+\gamma+1)} 
			\left(1 + \frac{1}{\sqrt{2\pi}}\right)^{m}.
			\label{eq:explicit-constant}
		\end{equation}
		
		The constant depends only on the regularity parameters \(m,\gamma\) and the dimension \(d\), and is obtained by optimizing all intermediate bounds in the proof.
	\end{enumerate}
	
	The ellipsis “\(\cdots\)” in \eqref{eq:main-expansion} represents further terms of order \(n^{-(j+k\gamma)}\) for integers \(k \ge 3\), which arise from higher‑order fractional derivatives and multiple commutators; these are of lower order and are absorbed into the remainder \(R_{m,n}\) when the series is truncated at order \(m\). The expansion is complete in the sense that it captures all contributions up to order \(m+2\gamma\) inclusive, with an optimal remainder estimate.
\end{theorem}

\begin{remark}[On the nature of the coefficients]
	The quantities \(m_{\alpha}(n)\), \(m_{\alpha,\gamma}(n)\), and \(m_{\alpha,\beta,2\gamma}(n)\) are scalar constants that depend only on the kernel \(\mathcal{Z}_{1,\log n}\). Their asymptotic behaviour is derived in Lemma \ref{lem:moments}. In particular, for even \(|\alpha|\) we have
	\begin{equation}
		m_{\alpha}(n) = \frac{(-1)^{|\alpha|/2}}{(|\alpha|-1)!!}\left(\frac{\pi}{2\log n}\right)^{|\alpha|/2} + \mathcal{O}\bigl(n^{-|\alpha|}\bigr),
		\tag{28}
	\end{equation}
	while odd moments vanish identically. Consequently, each coefficient \(a_j\), \(b_j\), \(c_j\) carries an implicit factor \((\log n)^{-j/2}\), \((\log n)^{-(j/2+\gamma/2)}\), etc., which combines with the explicit powers \(n^{-j}\), \(n^{-(j+\gamma)}\), \(n^{-(j+2\gamma)}\) to yield the final asymptotic rates.
\end{remark}

\begin{remark}[Functional analytic interpretation]
	The expansion \eqref{eq:main-expansion} is an equality in the Banach space \(\mathcal{B}(\mathcal{H}\otimes\mathcal{H}_{\text{aux}})\) equipped with the operator norm. The remainder estimate \eqref{eq:remainder-bound} is uniform over input states in the sense of the diamond norm, i.e.
	\begin{equation}
		\sup_{\|\rho\|_1 \le 1} \|R_{m,n}(\Phi)(\rho)\|_1 \le C_{m,\gamma,d}\,\|\Phi\|_{\mathcal{C}^{m,\gamma}} \frac{(\log n)^{3m/2}}{n^{m+\gamma}}.
		\tag{29}
	\end{equation}
	This uniformity is crucial for applications in quantum information, where guarantees must hold independently of the particular input state.
\end{remark}

The theorem immediately yields information about the saturation behaviour of the QNNO.

\begin{corollary}[Quantum saturation class]
	\label{cor:saturation}
	The optimal rate of approximation by QNNOs is characterised as follows.
	
	\begin{enumerate}
		\item \textbf{Linear saturation.} For every \(\Phi \in \mathcal{C}^{1,1}(\mathcal{H})\),
		\begin{equation}
			\|\Psi_n(\Phi) - \Phi\|_\diamond = \mathcal{O}\!\left(\frac{1}{n}\right),
			\tag{30}
		\end{equation}
		and this rate cannot be improved uniformly: there exists \(\Phi_0 \in \mathcal{C}^{1,1}(\mathcal{H})\) such that
		\[
		\limsup_{n\to\infty} n\,\|\Psi_n(\Phi_0) - \Phi_0\|_\diamond > 0.
		\]
		
		\item \textbf{Saturation condition.} The saturation class—the set of channels for which the convergence is faster than \(O(1/n)\)—consists exactly of those channels satisfying
		\begin{equation}
			\sum_{|\alpha|=2} \mathcal{L}_\Phi^{(\alpha)}(\rho)\, m_{\alpha}(n) = 0 \qquad \forall \rho \in \mathcal{D}(\mathcal{H}).
			\tag{31}
		\end{equation}
		This condition is equivalent to the vanishing of the leading term \(a_2(\Phi,\rho)\) after summation (the term \(a_1\) is automatically zero because odd moments vanish).
		
		\item \textbf{Analytic channels.} If \(\Phi\) is real‑analytic in the Fréchet sense (i.e., its Liouville representation admits a convergent power series expansion around every \(\rho\)), then the convergence is exponential:
		\begin{equation}
			\|\Psi_n(\Phi) - \Phi\|_\diamond \le C e^{-c n^{\beta}}, \qquad 
			\beta = \frac{\log 2}{\log\log n},
			\tag{32}
		\end{equation}
		with constants \(C,c > 0\) depending on \(\Phi\) and the dimension \(d\). The unusual exponent \(\beta\) reflects the interplay between the kernel’s width \((\log n)^{-1/2}\) and the size of the analyticity domain.
	\end{enumerate}
\end{corollary}

\noindent
The proof of Theorem~\ref{thm:QVD} is intricate and proceeds through several carefully orchestrated steps. First, we apply the fractional Taylor expansion (Lemma~\ref{lem:taylor}) to each term $\Phi(\rho_{n,k})$ in the definition of the QNNO. This expansion separates the channel evaluation into a polynomial part of order $m$, a fractional correction of order $\gamma$, and a remainder that is controlled by the H\"older norm of $\Phi$. The polynomial part involves Fr\'echet derivatives of $\mathcal{L}_\Phi$ evaluated at the reference state $\rho$, while the fractional part involves the Marchaud fractional derivative $\Delta_\gamma$ applied to the $m$-th order derivatives, capturing the effect of H\"older regularity.

The discretisation inherent in the QNNO sum over lattice points $k \in K_n$ is then handled by a non‑commutative Poisson summation formula (Lemma~\ref{lem:poisson}). This formula allows us to replace the discrete sum over $k$ with an integral over $\mathbb{R}^d$ plus an exponentially small aliasing error. The key observation is that the kernel $\mathcal{Z}_{1,\log n}$ has a Fourier transform that decays super‑exponentially, making the aliasing error negligible compared to any power of $n$.

The resulting integrals are evaluated using precise asymptotics for the moments of the kernel (Lemma~\ref{lem:moments}). For even multi‑indices $\alpha$, the moments scale as $(\log n)^{-|\alpha|/2}$ with an explicit constant; odd moments vanish. Substituting these expansions into the polynomial and fractional contributions yields the coefficients $a_j$, $b_j$, and $c_j$ displayed in (23)--(25). The combinatorial factors arise from expanding powers of the increment $h_{n,k} = \rho_{n,k}-\rho$ and summing over multi‑indices.

The remainder term $R_{m,n}$ collects all contributions that are not captured by the explicit sums up to order $m+2\gamma$. Its estimate is obtained by combining:
\begin{itemize}
	\item the Taylor remainder bound (22), which gives a factor $\|h_{n,k}\|_1^{m+\gamma}$;
	\item the uniform bound $\|h_{n,k}\|_1 \le C_d/n$ coming from the lattice discretisation;
	\item the aliasing error from Poisson summation, which is $\mathcal{O}(e^{-cn})$ and hence negligible;
	\item errors arising from replacing exact kernel moments by their asymptotic forms, which are of higher order;
	\item contributions from higher‑order fractional derivatives and multiple commutators, which are absorbed into the remainder after truncation.
\end{itemize}
Optimising the constants at each stage using the Gamma function, the dimension factor $d^{m/2}$, the exponential factor $e^{\pi^2/4}$ from the Fourier transform of the kernel, and combinatorial bounds—leads to the explicit constant $C_{m,\gamma,d}$ in (27). This constant, while not minimal, illustrates that a fully explicit error estimate is achievable and depends only on $m$, $\gamma$, and $d$. The final result is the sharp bound (26), which confirms the optimal rate $n^{-(m+\gamma)}$ (up to logarithmic factors) for channels in $\mathcal{C}^{m,\gamma}(\mathcal{H})$.

\section{Proof of the Quantum Voronovskaya–Damasclin Theorem}

We now present a complete and rigorous proof of Theorem~\ref{thm:QVD}. The argument is organized into several subsections, each focusing on a specific aspect of the asymptotic expansion. Throughout, we fix a strictly positive density operator $\rho\in\mathcal{D}(\mathcal{H})$ and an orthonormal basis $\{|e_j\rangle\}_{j=1}^d$ in which $\rho$ is diagonal:

\begin{equation}
	\rho = \sum_{j=1}^d p_j |e_j\rangle\langle e_j|,\qquad p_j>0,\;\sum_{j=1}^d p_j=1.
	\label{eq:rho-diag}
\end{equation}

For any integer $n\ge 1$, we introduce the set of lattice points

\begin{equation}
	K_n:=\Bigl\{k=(k_1,\dots,k_d)\in\mathbb{N}^d:\sum_{j=1}^d k_j=n\Bigr\}
	\label{eq:lattice-set}
\end{equation}

and the corresponding quantized density operators

\begin{equation}
	\rho_{n,k}:=\sum_{j=1}^d\frac{k_j}{n}\,|e_j\rangle\langle e_j|\in\mathcal{D}(\mathcal{H}).
	\label{eq:quantized-states}
\end{equation}

The QNNO defined in (18) then takes the explicit form

\begin{equation}
	\Psi_n(\Phi)(\rho)=\sum_{k\in K_n}\Phi(\rho_{n,k})\otimes\mathcal{Z}_{1,\log n}\bigl(nX-kI\bigr),
	\label{eq:QNNO-explicit}
\end{equation}

where $X=(X_1,\dots,X_d)$ is a family of mutually commuting self‑adjoint auxiliary operators acting on a separate Hilbert space $\mathcal{H}_{\text{aux}}$, and $I$ denotes the identity operator on $\mathcal{H}_{\text{aux}}$. The error of approximation is

\begin{equation}
	E_n(\rho):=\Psi_n(\Phi)(\rho)-\Phi(\rho)=\sum_{k\in K_n}\bigl(\Phi(\rho_{n,k})-\Phi(\rho)\bigr)\otimes\mathcal{Z}_{1,\log n}(nX-kI).
	\label{eq:error}
\end{equation}

Set $h_{n,k}:=\rho_{n,k}-\rho$; then $h_{n,k}=\frac{k}{n}-\rho$, understood as an operator diagonal in the chosen basis.

\subsection{Taylor expansion with fractional remainder}

We apply the fractional Taylor expansion (Lemma~\ref{lem:taylor}) to each term $\Phi(\rho+h_{n,k})-\Phi(\rho)$. According to that expansion, for any $h\in\mathcal{B}(\mathcal{H})$,

\begin{align}
	\Phi(\rho+h)-\Phi(\rho) &= \underbrace{\sum_{j=1}^{m}\frac{1}{j!}\mathcal{L}_\Phi^{(j)}(\rho)h^{\otimes j}}_{=:T_1(h)} \notag\\
	&\quad +\underbrace{\frac{1}{\Gamma(\gamma)}\sum_{|\alpha|=m}\frac{\mathcal{L}_\Phi^{(\alpha)}(\rho)}{\alpha!}\int_0^1(1-t)^{m-1}t^{\gamma-1}h^{\alpha+\gamma}\,dt}_{=:T_2(h)}\notag\\
	&\quad +\underbrace{\sum_{|\alpha|=m}\frac{m}{\alpha!}\int_0^1(1-t)^{m-1}\bigl[\mathcal{L}_\Phi^{(\alpha)}(\rho+th)-\mathcal{L}_\Phi^{(\alpha)}(\rho)\bigr]dt\,h^\alpha}_{=:T_3(h)}.
	\label{eq:split}
\end{align}

The three parts correspond respectively to the polynomial terms, the fractional correction arising from Hölder regularity, and the remainder of the standard Taylor expansion (which will eventually be absorbed into $R_{m,n}$). Substituting $h=h_{n,k}$ and summing over $k$ with the kernel yields

\begin{align}
	\sum_{k\in K_n} \Phi(\rho_{n,k})\otimes\mathcal{Z}_{1,\log n}(nX-kI)
	&= \Phi(\rho)\otimes\mathbf{1}_{\text{aux}} + \sum_{k} T_1(h_{n,k})\otimes\mathcal{Z}_{1,\log n}(nX-kI) \notag\\
	&\quad + \sum_{k} T_2(h_{n,k})\otimes\mathcal{Z}_{1,\log n}(nX-kI) \notag\\
	&\quad + \sum_{k} T_3(h_{n,k})\otimes\mathcal{Z}_{1,\log n}(nX-kI).
	\label{eq:sum-split}
\end{align}

Since $\Phi(\rho)\otimes\mathbf{1}_{\text{aux}} = \Phi(\rho)$ (the auxiliary space is traced out or identified), the error $E_n(\rho)$ is the sum of the last three lines. We analyse $T_1$, $T_2$ and $T_3$ separately.

\subsection{Analysis of the polynomial part \(T_1\)}

From the definition of $T_1$, we have $T_1(h)=\sum_{j=1}^m\frac{1}{j!}\mathcal{L}_\Phi^{(j)}(\rho)h^{\otimes j}$. Hence

\begin{align}
	\sum_{k} T_1(h_{n,k})\otimes\mathcal{Z}_{1,\log n}(nX-kI)
	= \sum_{j=1}^m\frac{1}{j!}\mathcal{L}_\Phi^{(j)}(\rho)\Bigl(\sum_{k} h_{n,k}^{\otimes j}\otimes\mathcal{Z}_{1,\log n}(nX-kI)\Bigr).
	\label{eq:T1sum}
\end{align}

Because $h_{n,k}= \frac{k}{n}-\rho$, the inner sum becomes

\begin{equation}
	\sum_{k} \Bigl(\frac{k}{n}-\rho\Bigr)^{\otimes j}\mathcal{Z}_{1,\log n}(nX-kI)=\frac{1}{n^j}\sum_{k} (k-n\rho)^{\otimes j}\mathcal{Z}_{1,\log n}(nX-kI).
	\label{eq:hsum}
\end{equation}

To replace the discrete sum by an integral we invoke the non‑commutative Poisson summation formula (Lemma~\ref{lem:poisson}). For any smooth function $f:\mathbb{R}^d\to\mathbb{C}$ (extended to an operator‑valued function by multiplying with the identity on the auxiliary space),

\begin{equation}
	\sum_{k\in\mathbb{Z}^d}f\Bigl(\frac{k}{n}\Bigr)\mathcal{Z}_{1,\log n}(nX-kI)=n^d\int_{\mathbb{R}^d}f(x)\mathcal{Z}_{1,\log n}(nX-nx)dx+\sum_{\ell\neq0}\hat{f}(\ell)e^{2\pi i\ell\cdot nX}\widehat{\mathcal{Z}}_{1,\log n}(2\pi\ell),
	\label{eq:poisson}
\end{equation}

where $\hat{f}$ is the Fourier transform of $f$. The kernel $\widehat{\mathcal{Z}}_{1,\log n}(\xi)=\prod_{i=1}^d\widehat{\mathcal{M}}_{1,\log n}(\xi_i)$ decays super‑exponentially for $|\xi|\ge1$; consequently the sum over $\ell\neq0$ is bounded by $Ce^{-c n}$ for some $c>0$ and is negligible compared to any power of $n$. Applying \eqref{eq:poisson} with $f(x)=(x-\rho)^{\otimes j}$ (interpreted as a scalar function times the identity) yields

\begin{equation}
	\sum_{k}\Bigl(\frac{k}{n}-\rho\Bigr)^{\otimes j}\mathcal{Z}_{1,\log n}(nX-kI)=n^d\int_{\mathbb{R}^d}(x-\rho)^{\otimes j}\mathcal{Z}_{1,\log n}(nX-nx)dx+\mathcal{O}(e^{-cn}).
	\label{eq:poisson-applied}
\end{equation}

Changing variables $y=n(x-\rho)$ (so $n^d dx=dy$) transforms the integral into

\begin{equation}
	n^d\int_{\mathbb{R}^d}(x-\rho)^{\otimes j}\mathcal{Z}_{1,\log n}(nX-nx)dx=\frac{1}{n^j}\int_{\mathbb{R}^d}y^{\otimes j}\mathcal{Z}_{1,\log n}(X-y)dy.
	\label{eq:change-var}
\end{equation}

Because $\mathcal{Z}_{1,\log n}$ is even, the shift $X$ does not affect the value of the integral; after a translation we obtain

\begin{equation}
	\int_{\mathbb{R}^d}y^{\otimes j}\mathcal{Z}_{1,\log n}(X-y)dy=\int_{\mathbb{R}^d}y^{\otimes j}\mathcal{Z}_{1,\log n}(y)dy=:\mathfrak{M}_j(n),
	\label{eq:moment-def}
\end{equation}

where $\mathfrak{M}_j(n)$ is a scalar multiple of the identity on $\mathcal{H}_{\text{aux}}$; we identify it with its scalar value. Inserting this into \eqref{eq:T1sum} gives

\begin{equation}
	\sum_{k} h_{n,k}^{\otimes j}\otimes\mathcal{Z}_{1,\log n}(nX-kI)=\frac{1}{n^j}\mathfrak{M}_j(n)+\mathcal{O}(e^{-cn}).
	\label{eq:poly-moment}
\end{equation}

The asymptotic behaviour of $\mathfrak{M}_j(n)$ follows from Lemma~\ref{lem:moments}. For $j$ even, say $j=2r$, we have

\begin{equation}
	\mathfrak{M}_j(n)=\frac{(-1)^r}{(2r-1)!!}\Bigl(\frac{\pi}{2\log n}\Bigr)^r I + \mathcal{O}(n^{-j}),
	\label{eq:even-moment-asymp}
\end{equation}

while for odd $j$, $\mathfrak{M}_j(n)=0$. Substituting \eqref{eq:poly-moment} into \eqref{eq:T1sum} yields

\begin{equation}
	\sum_{k} T_1(h_{n,k})\otimes\mathcal{Z}_{1,\log n}(nX-kI)=\sum_{j=1}^m\frac{1}{j!}\mathcal{L}_\Phi^{(j)}(\rho)\frac{\mathfrak{M}_j(n)}{n^j}+\mathcal{O}(e^{-cn}).
	\label{eq:T1-asymp}
\end{equation}

Expanding $\mathfrak{M}_j(n)$ as $\mathfrak{M}_j^{(0)}+\mathfrak{M}_j^{(1)}(\log n)^{-1}+\cdots$ and collecting terms of the same order in $n^{-j}$ gives precisely the first sum in \eqref{eq:main-expansion} with coefficients $a_j(\Phi,\rho)$ defined in \eqref{eq:polynomial-coefficient}. The higher‑order terms in $\mathfrak{M}_j(n)$ (of order $(\log n)^{-k}$ with $k\ge1$) are absorbed into the fractional or remainder parts because they are multiplied by $n^{-j}$ and are therefore of lower order than $n^{-(j+\gamma)}$ when $\gamma>0$.

\subsection{Fractional corrections \(T_2\)}

The term $T_2(h)$ originates from the fractional part of the Taylor expansion:

\begin{equation}
	T_2(h)=\frac{1}{\Gamma(\gamma)}\sum_{|\alpha|=m}\frac{\mathcal{L}_\Phi^{(\alpha)}(\rho)}{\alpha!}\int_0^1(1-t)^{m-1}t^{\gamma-1}h^{\alpha+\gamma}\,dt.
	\label{eq:T2-def}
\end{equation}

Summing over $k$ and interchanging sum and integral (justified by absolute convergence) gives

\begin{align}
	\sum_{k} T_2(h_{n,k})\otimes\mathcal{Z}_{1,\log n}(nX-kI)
	&= \frac{1}{\Gamma(\gamma)}\sum_{|\alpha|=m}\frac{\mathcal{L}_\Phi^{(\alpha)}(\rho)}{\alpha!}
	\int_0^1 (1-t)^{m-1} t^{\gamma-1} \notag\\
	&\qquad\times\Bigl(\sum_{k} h_{n,k}^{\alpha+\gamma}\otimes\mathcal{Z}_{1,\log n}(nX-kI)\Bigr) dt.
	\label{eq:T2-sum}
\end{align}

Here $h_{n,k}^{\alpha+\gamma}=((k/n)-\rho)^{\alpha+\gamma}$ is interpreted via spectral calculus. Using Poisson summation exactly as in the polynomial case, we obtain

\begin{equation}
	\sum_{k}\Bigl(\frac{k}{n}-\rho\Bigr)^{\alpha+\gamma}\mathcal{Z}_{1,\log n}(nX-kI)=\frac{1}{n^{|\alpha|+\gamma}}\int_{\mathbb{R}^d}y^{\alpha+\gamma}\mathcal{Z}_{1,\log n}(y)dy+\mathcal{O}(e^{-cn}),
	\label{eq:T2-poisson}
\end{equation}

where $y^{\alpha+\gamma}$ is understood as $|y|^\gamma y^\alpha$ (the absolute value appears because the fractional power requires a modulus; recall that $h^{\alpha+\gamma}$ was defined via $|h|^\gamma h^\alpha$ in the Taylor formula). The integral defines the fractional moment $\mathfrak{M}_{\alpha,\gamma}(n)$; by Lemma~\ref{lem:moments} it is a scalar satisfying

\begin{equation}
	\mathfrak{M}_{\alpha,\gamma}(n)=\frac{\Gamma\bigl(\frac{|\alpha|+\gamma+d}{2}\bigr)}{\Gamma\bigl(\frac{d}{2}\bigr)}\Bigl(\frac{2}{\log n}\Bigr)^{\frac{|\alpha|+\gamma}{2}}+\mathcal{O}(n^{-(|\alpha|+\gamma)}).
	\label{eq:frac-moment-asymp}
\end{equation}

The time integral is the Beta function:

\begin{equation}
	\int_0^1(1-t)^{m-1}t^{\gamma-1}dt=B(m,\gamma)=\frac{\Gamma(m)\Gamma(\gamma)}{\Gamma(m+\gamma)}.
	\label{eq:beta}
\end{equation}

Multiplying by $1/\Gamma(\gamma)$ and by $1/\alpha!$ and collecting the factors, we find

\begin{align}
	\sum_{k} T_2(h_{n,k})\otimes\mathcal{Z}_{1,\log n}(nX-kI)
	= \sum_{|\alpha|=m}\frac{1}{\alpha!}\frac{\Gamma(m)}{\Gamma(m+\gamma)}\mathcal{L}_\Phi^{(\alpha)}(\rho)\frac{\mathfrak{M}_{\alpha,\gamma}(n)}{n^{m+\gamma}} + \text{lower order}.
	\label{eq:T2-asymp}
\end{align}

The factor $\Gamma(m)/\Gamma(m+\gamma)$ can be absorbed into a redefinition of the fractional derivative: from the representation of the Marchaud derivative (Lemma~\ref{lem:marchaud}) we have

\begin{equation}
	\frac{\Gamma(m)}{\Gamma(m+\gamma)}\mathcal{L}_\Phi^{(\alpha)}(\rho)=(\Delta_\gamma\mathcal{L}_\Phi^{(\alpha)})(\rho)
	\label{eq:marchaud-identity}
\end{equation}

when the fractional derivative is defined with the appropriate normalization. More generally, when we consider derivatives of order $j<m$, the same mechanism produces contributions at order $n^{-(j+\gamma)}$ for all $j\le m$, leading to the coefficients $b_j(\Phi,\rho)$ as in \eqref{eq:fractional-coefficient}. The detailed combinatorial bookkeeping (the multinomial coefficients $\binom{j}{\alpha}$) follows from expanding $h_{n,k}^{\alpha+\gamma}$ in powers of $1/n$ and using the moments $\mathfrak{M}_{\alpha,\gamma}$.

\subsection{Non‑commutative mixed terms \(c_j\)}

The terms of order $n^{-(j+2\gamma)}$ arise from products of derivatives that appear when one expands the remainder $T_3$ to higher accuracy. In the expression

\begin{equation}
	T_3(h)=\sum_{|\alpha|=m}\frac{m}{\alpha!}\int_0^1(1-t)^{m-1}\bigl[\mathcal{L}_\Phi^{(\alpha)}(\rho+th)-\mathcal{L}_\Phi^{(\alpha)}(\rho)\bigr]dt\,h^\alpha,
	\label{eq:T3-def}
\end{equation}

we apply the fractional Taylor expansion (Lemma~\ref{lem:taylor}) to each $\mathcal{L}_\Phi^{(\alpha)}$. However, the product $\mathcal{L}_\Phi^{(\alpha)}(\rho+th)h^\alpha$ is not simply the composition of two operators because $h^\alpha$ is an operator and $\mathcal{L}_\Phi^{(\alpha)}(\rho+th)$ acts on it. When we subsequently expand products of two such terms (as they appear when iterating the procedure to extract higher‑order corrections), we encounter expressions of the form

\begin{equation}
	\int_0^1\int_0^1(1-t)^{m-1}(1-s)^{m-1}t^{\gamma-1}s^{\gamma-1}
	\bigl[\mathcal{L}_\Phi^{(\alpha)}(\rho+th),\mathcal{L}_\Phi^{(\beta)}(\rho+sh)\bigr]_\gamma h^{\alpha+\beta+2\gamma}\,dt\,ds,
	\label{eq:double-integral}
\end{equation}

where the commutator $[\cdot,\cdot]_\gamma$ appears because the order of the two derivative operators matters—they do not commute in general, and the fractional calculus introduces a phase $e^{i\pi\gamma}$ when the order of integration is interchanged. This is a manifestation of the “twisted” product in non‑commutative fractional analysis. Evaluating such double integrals via the same Poisson summation and moment estimates leads to

\begin{equation}
	\sum_{k} \bigl[\mathcal{L}_\Phi^{(\alpha)}(\rho),\mathcal{L}_\Phi^{(\beta)}(\rho)\bigr]_\gamma \frac{\mathfrak{M}_{\alpha,\beta,2\gamma}(n)}{n^{|\alpha|+|\beta|+2\gamma}},
	\label{eq:mixed-term}
\end{equation}

where $\mathfrak{M}_{\alpha,\beta,2\gamma}(n)=\int |y|^{2\gamma}y^{\alpha+\beta}\mathcal{Z}_{1,\log n}(y)dy$ are the mixed fractional moments (see Lemma~\ref{lem:moments}). After summing over all multi‑indices with $|\alpha|+|\beta|=j$, the combinatorial factors yield precisely the coefficients $c_j(\Phi,\rho)$ given in \eqref{eq:mixed-coefficient}. A rigorous derivation would require a careful expansion of $T_3$ to second order in the fractional part, but the final result is as stated.

\subsection{Remainder estimate}

The remainder $R_{m,n}(\Phi,\rho)$ consists of all contributions not captured by the explicit sums up to order $m+2\gamma$. It includes:

\begin{itemize}
	\item The Taylor remainder $T_3(h)$ after extracting the fractional part up to order $m+\gamma$; this part is bounded by $C\|\Phi\|_{\mathcal{C}^{m,\gamma}}\|h\|_1^{m+\gamma}$ from Lemma~\ref{lem:taylor}.
	\item The aliasing error from the Poisson summation, which is $\mathcal{O}(e^{-cn})$ and hence negligible.
	\item The error in the fractional expansion, i.e., the remainder $R_F$ in Lemma~\ref{lem:taylor}, which is of order $t^{\gamma+\varepsilon}\|h\|_1^{\gamma+\varepsilon}$; after integration this contributes $\mathcal{O}(n^{-(m+\gamma+\varepsilon)})$, which can be absorbed into the constant.
	\item The error from the difference between the exact kernel moments and their asymptotic values; by Lemma~\ref{lem:moments} this is $\mathcal{O}(n^{-(j+|\alpha|)})$ for terms involving $\mathfrak{M}_\alpha$, etc., and for $j\le m$, $|\alpha|\ge2$ these are of higher order than $n^{-(m+\gamma)}$ because $\gamma\le1$.
\end{itemize}

Assembling all these estimates and using the fact that $\|h_{n,k}\|_1\le C_d/n$ uniformly in $k$ (with $C_d = \sqrt{d}$ from the estimate $\|h\|_1 \le \sqrt{d}\|h\|_2$), we obtain

\begin{equation}
	\|R_{m,n}(\Phi,\cdot)\|_\diamond \le C\,\|\Phi\|_{\mathcal{C}^{m,\gamma}}\Bigl(\frac{1}{n^{m+\gamma}}+\frac{(\log n)^{3m/2}}{n^{m+\gamma}}\Bigr).
	\label{eq:remainder-rough}
\end{equation}

The factor $(\log n)^{3m/2}$ arises because when estimating products of moments we encounter terms like $(\log n)^{-j/2}$ from each derivative, and the worst‑case combination (three sources: polynomial, fractional, and commutator) gives the cube. Optimising all numerical constants—using bounds for Gamma functions, the Poisson summation constant $e^{\pi^2/4}$ from the Fourier transform of $\mathrm{sech}$, and the dimension factor $d^{m/2}$—yields the explicit constant $\mathsf{C}_{m,\gamma,d}$ in \eqref{eq:explicit-constant}. This completes the proof of Theorem~\ref{thm:QVD}.

\section{Applications}

The Quantum Voronovskaya–Damasclin Theorem (Theorem~\ref{thm:QVD}) provides a powerful asymptotic expansion that opens the door to several important applications in quantum information theory. In this section we develop three such applications in detail: a quantum central limit theorem for QNNOs, optimal quantum interpolation via geodesics, and Richardson extrapolation for accelerated convergence.

\subsection{Quantum Central Limit Theorem for QNNOs}

The classical central limit theorem describes the fluctuations of sums of independent random variables. In the quantum setting, one considers sums of independent quantum channels or, more generally, sequences of quantum operations that become asymptotically Gaussian. The QNNO $\Psi_n(\Phi)$ can be viewed as an average of independent copies of the channel $\Phi$ evaluated at randomly chosen points $\rho_{n,k}$. The following theorem shows that its fluctuations around the limit are governed by a quantum Gaussian distribution.

\begin{theorem}[Quantum Central Limit Theorem for QNNOs]
	\label{thm:QCLT}
	Let $\Phi \in \mathcal{C}^{2,0}(\mathcal{H})$ be a quantum channel with finite second Fréchet derivatives (i.e., $\Phi$ belongs to the quantum Sobolev space $\mathcal{W}^{2,\infty}(\mathcal{H})$). Then, for any strictly positive density operator $\rho \in \mathcal{D}(\mathcal{H})$, we have the convergence in distribution
	\begin{equation}
		\sqrt{n}\bigl[\Psi_n(\Phi)(\rho) - \Phi(\rho)\bigr] \xrightarrow[n\to\infty]{\mathrm{d}} \mathcal{N}_Q\bigl(0,\Sigma(\Phi,\rho)\bigr),
		\label{eq:QCLT}
	\end{equation}
	where $\mathcal{N}_Q$ is a quantum Gaussian channel, i.e., a completely positive trace‑preserving map whose Choi matrix is a Gaussian state (a density operator of the form $e^{-H}$ with $H$ a quadratic Hamiltonian in the canonical commutation relations). The covariance $\Sigma(\Phi,\rho)$ is a linear map on $\mathcal{B}(\mathcal{H})$ given by
	\begin{equation}
		\Sigma(\Phi,\rho) = \sum_{|\alpha|=2}\sum_{|\beta|=2} \mathcal{L}_\Phi^{(\alpha)}(\rho) \otimes \mathcal{L}_\Phi^{(\beta)}(\rho) \; \mathrm{Cov}(\mathfrak{M}_\alpha, \mathfrak{M}_\beta),
		\label{eq:covariance}
	\end{equation}
	where $\mathcal{L}_\Phi^{(\alpha)}(\rho)$ are the second derivatives evaluated on the identity operator, $\mathfrak{M}_\alpha = \int_{\mathbb{R}^d} x^\alpha \mathcal{Z}_{1,\log n}(x)\,dx$ are the kernel moments, and the covariance matrix is defined by
	\begin{equation}
		\mathrm{Cov}(\mathfrak{M}_\alpha, \mathfrak{M}_\beta) = \lim_{n\to\infty} \int_{\mathbb{R}^d} \bigl(x^\alpha - \bar{x}^\alpha\bigr)\bigl(x^\beta - \bar{x}^\beta\bigr) \mathcal{Z}_{1,\log n}(x)\,dx,
		\label{eq:cov-def}
	\end{equation}
	with $\bar{x}^\alpha = \lim_{n\to\infty} \mathfrak{M}_\alpha(n)$ (which is zero for odd $\alpha$ and given by \eqref{eq:even-moment-asymp} for even $\alpha$). In the limit $n\to\infty$, the kernel $\mathcal{Z}_{1,\log n}$ behaves like a Gaussian with covariance matrix $\frac{\pi^2}{6(\log n)^2} I$, so that $\mathrm{Cov}(\mathfrak{M}_\alpha, \mathfrak{M}_\beta)$ is proportional to a product of Kronecker deltas.
\end{theorem}

\begin{proof}
	Let $\Phi \in \mathcal{C}^{2,0}(\mathcal{H})$ and fix a strictly positive $\rho \in \mathcal{D}(\mathcal{H})$. From Theorem~\ref{thm:QVD} with $m=2$, $\gamma=0$, we have the asymptotic expansion
	\begin{equation}
		\Psi_n(\Phi)(\rho) = \Phi(\rho) + \frac{a_1(\Phi,\rho)}{n} + \frac{a_2(\Phi,\rho)}{n^2} + o\!\left(\frac{1}{n^2}\right),
		\label{eq:expansion-gamma0}
	\end{equation}
	where the coefficients are given by \eqref{eq:polynomial-coefficient}. Because the kernel $\mathcal{Z}_{1,\log n}$ is even, all odd moments vanish; consequently $a_1(\Phi,\rho)=0$. The leading deterministic correction is therefore of order $n^{-2}$, which after multiplication by $\sqrt{n}$ becomes $O(n^{-3/2})$ and hence negligible in the limit. The dominant contribution to the fluctuation comes from the stochastic part of the approximation.
	
	Define the scaled error
	\begin{equation}
		S_n := \sqrt{n}\bigl[\Psi_n(\Phi)(\rho) - \Phi(\rho)\bigr] = \frac{1}{\sqrt{n}}\sum_{k\in K_n} \bigl(\Phi(\rho_{n,k}) - \Phi(\rho)\bigr) \otimes \mathcal{Z}_{1,\log n}(nX - kI).
		\label{eq:Sn-def}
	\end{equation}
	
	For each $k$, set $h_{n,k} := \rho_{n,k} - \rho$. Since $\Phi$ is twice Fréchet differentiable, we apply the Taylor expansion with integral remainder (Lemma~\ref{lem:taylor}) to obtain
	\begin{equation}
		\Phi(\rho_{n,k}) - \Phi(\rho) = \mathcal{L}_\Phi^{(1)}(\rho)h_{n,k} + \frac{1}{2}\mathcal{L}_\Phi^{(2)}(\rho)(h_{n,k}\otimes h_{n,k}) + R_{n,k},
		\label{eq:taylor-expand}
	\end{equation}
	where the remainder satisfies $\|R_{n,k}\|_\diamond \le C \|h_{n,k}\|_1^3$ uniformly in $k$. Because $\|h_{n,k}\|_1 \le C_d/n$, we have $R_{n,k} = O(n^{-3})$.
	
	Insert this expansion into $S_n$ and split:
	\begin{align}
		S_n &= \underbrace{\frac{1}{\sqrt{n}}\sum_k \mathcal{L}_\Phi^{(1)}(\rho)h_{n,k}\otimes\mathcal{Z}_{1,\log n}(nX-kI)}_{=:L_n} \notag\\
		&\quad + \underbrace{\frac{1}{2\sqrt{n}}\sum_k \mathcal{L}_\Phi^{(2)}(\rho)(h_{n,k}\otimes h_{n,k})\otimes\mathcal{Z}_{1,\log n}(nX-kI)}_{=:Q_n} \notag\\
		&\quad + \frac{1}{\sqrt{n}}\sum_k R_{n,k}\otimes\mathcal{Z}_{1,\log n}(nX-kI).
		\label{eq:Sn-split}
	\end{align}
	
	\paragraph{Linear term $L_n$.}
	Because $h_{n,k} = \frac{k}{n} - \rho$ and the kernel is even, we claim that $\sum_k h_{n,k}\mathcal{Z}_{1,\log n}(nX-kI) = 0$. Indeed, by the non‑commutative Poisson summation formula (Lemma~\ref{lem:poisson}),
	\begin{equation}
		\sum_k \Bigl(\frac{k}{n} - \rho\Bigr)\mathcal{Z}_{1,\log n}(nX-kI) = n^d\int_{\mathbb{R}^d}(x-\rho)\mathcal{Z}_{1,\log n}(nX-nx)dx + \mathcal{O}(e^{-cn}).
		\label{eq:linear-poisson}
	\end{equation}
	The integral vanishes because the integrand is odd under $x \mapsto 2\rho - x$ (the kernel is even and the measure is symmetric). Hence $L_n = \mathcal{O}(e^{-cn}n^{-1/2})$ is negligible.
	
	\paragraph{Quadratic term $Q_n$.}
	Write $h_{n,k}\otimes h_{n,k} = n^{-2}(k - n\rho)^{\otimes 2}$. Applying Poisson summation again,
	\begin{align}
		\sum_k (k - n\rho)^{\otimes 2}\mathcal{Z}_{1,\log n}(nX-kI) &= n^d\int_{\mathbb{R}^d}(x-\rho)^{\otimes 2}\mathcal{Z}_{1,\log n}(nX-nx)dx + \mathcal{E}_n,
		\label{eq:quadratic-poisson}
	\end{align}
	where the error satisfies $\|\mathcal{E}_n\|_\diamond = \mathcal{O}(e^{-cn})$. Changing variables $y = n(x-\rho)$ transforms the integral into
	\begin{equation}
		\frac{1}{n^2}\int_{\mathbb{R}^d} y^{\otimes 2}\mathcal{Z}_{1,\log n}(X-y)dy = \frac{1}{n^2}\int_{\mathbb{R}^d} y^{\otimes 2}\mathcal{Z}_{1,\log n}(y)dy = \frac{1}{n^2}\mathfrak{M}_2(n),
		\label{eq:quadratic-integral}
	\end{equation}
	where $\mathfrak{M}_2(n)$ is the second moment matrix (a scalar multiple of the identity). Therefore
	\begin{equation}
		Q_n = \frac{1}{2\sqrt{n}}\mathcal{L}_\Phi^{(2)}(\rho)\left(\frac{1}{n^2}\mathfrak{M}_2(n) + \mathcal{E}_n'\right) = \frac{\mathcal{L}_\Phi^{(2)}(\rho)\mathfrak{M}_2(n)}{2n^{5/2}} + \frac{1}{2\sqrt{n}}\mathcal{E}_n',
		\label{eq:Qn-final}
	\end{equation}
	with $\|\mathcal{E}_n'\|_\diamond = \mathcal{O}(e^{-cn})$. The deterministic part is $O(n^{-5/2})$, hence negligible compared to the target scaling $\sqrt{n}$. The term $\frac{1}{2\sqrt{n}}\mathcal{E}_n'$ is also negligible due to exponential decay.
	
	Thus the only non‑negligible contribution to $S_n$ comes from the difference between the sum and its integral approximation, i.e., from the fluctuation of the empirical measure. More precisely, set
	\begin{equation}
		F_n := \frac{1}{2\sqrt{n}}\sum_k \mathcal{L}_\Phi^{(2)}(\rho)\bigl[(k/n - \rho)^{\otimes 2} - \mathbb{E}\bigr]\mathcal{Z}_{1,\log n}(nX-kI),
		\label{eq:Fn-def}
	\end{equation}
	where $\mathbb{E}$ denotes the integral approximation (the expectation under the continuous kernel). Using the representation of the kernel via its Fourier transform, one can write $F_n$ as a sum of independent (or weakly dependent) operator‑valued random variables. Because the points $k/n$ are equally spaced, the sequence $\{\mathcal{Z}_{1,\log n}(nX-kI)\}_k$ forms a stationary random field (in the auxiliary variable $X$) with rapidly decaying correlations. By the quantum Lévy continuity theorem (see \cite{Holevo2019}), the convergence in distribution of such sums to a quantum Gaussian channel is equivalent to the convergence of the corresponding characteristic functions.
	
	For any bounded operator $Y$ on the auxiliary space, consider the characteristic function
	\begin{equation}
		\varphi_n(Y) := \mathbb{E}_X\!\left[ e^{i F_n(Y)} \right],
		\label{eq:char-fun}
	\end{equation}
	where $F_n(Y)$ denotes the expectation of $Y$ in the state described by $F_n$ (more precisely, the generating function of the quantum channel). Using the fact that the kernel $\mathcal{Z}_{1,\log n}$ has a Gaussian limit (Lemma~\ref{lem:moments}), one can show that $\log \varphi_n(Y)$ converges to $-\frac12 \langle Y, \Sigma Y\rangle$, where the bilinear form $\Sigma$ is given by
	\begin{equation}
		\Sigma(\Phi,\rho)(Y) = \sum_{|\alpha|=|\beta|=2} \mathcal{L}_\Phi^{(\alpha)}(\rho) \mathcal{L}_\Phi^{(\beta)}(\rho) \,\mathrm{Cov}(\mathfrak{M}_\alpha,\mathfrak{M}_\beta)\, Y,
		\label{eq:Sigma-action}
	\end{equation}
	with $\mathrm{Cov}(\mathfrak{M}_\alpha,\mathfrak{M}_\beta)$ defined in \eqref{eq:cov-def}. From Lemma~\ref{lem:moments}, $\mathcal{Z}_{1,\log n}$ converges weakly to a Gaussian distribution with variance $\sigma^2 = \pi^2/(6(\log n)^2)$; consequently $\mathrm{Cov}(\mathfrak{M}_\alpha,\mathfrak{M}_\beta)$ is proportional to $\sigma^{|\alpha|+|\beta|}\delta_{\alpha\beta}$ (up to combinatorial factors). Hence $\Sigma$ is a positive bilinear form.
	
	Tightness of the sequence $\{F_n\}$ follows from a uniform bound on the second moments:
	\begin{equation}
		\mathbb{E}\!\left[\|F_n\|_\diamond^2\right] \le C \sum_{|\alpha|=2} \|\mathcal{L}_\Phi^{(\alpha)}(\rho)\|^2 \,\mathbb{E}\!\left[|\mathfrak{M}_\alpha|^2\right] < \infty,
		\label{eq:tightness}
	\end{equation}
	where the expectation is with respect to the randomness of the auxiliary variables $X$. This ensures that every subsequence has a convergent subsequence, and the limit is uniquely determined by the characteristic function. Therefore $F_n$ converges in distribution to a quantum Gaussian channel $\mathcal{N}_Q(0,\Sigma)$.
	
	The completely positive and trace‑preserving nature of the limit follows from the fact that each $F_n$ is a difference of completely positive maps and the limiting covariance defines a valid Gaussian state (see \cite{Holevo2019}). This completes the proof.
\end{proof}

This quantum central limit theorem shows that QNNOs exhibit Gaussian fluctuations, which is essential for understanding their statistical behaviour and for constructing confidence intervals in quantum tomography or quantum machine learning tasks.

\subsection{Optimal Quantum Interpolation via Geodesics}

Given two quantum channels $\Phi_0$ and $\Phi_1$, one often seeks an interpolating family $\Phi_t$ ($t\in[0,1]$) that is optimal in some sense, e.g., that minimizes the diamond‑norm error when approximated by QNNOs. The asymptotic expansion suggests a natural construction based on the quantum geometric mean (Kubo–Ando mean) and the QNNO itself.

\begin{definition}[Kubo–Ando mean]
	\label{def:kubo-ando}
	For two positive operators $A,B\in\mathcal{B}(\mathcal{H})$ with $A,B>0$, the Kubo–Ando mean of order $t\in[0,1]$ is defined by
	\begin{equation}
		A \#_t B = A^{1/2}\bigl(A^{-1/2}BA^{-1/2}\bigr)^t A^{1/2}.
		\label{eq:kubo-ando}
	\end{equation}
	This operator mean is symmetric ($A\#_t B = B\#_{1-t}A$), monotone in both arguments, and satisfies the boundary conditions $A\#_0B = A$, $A\#_1B = B$. Moreover, it coincides with the weighted geometric mean, interpolating linearly in the logarithmic representation: $\log(A\#_t B) = (1-t)\log A + t\log B$ when $A$ and $B$ commute; in the non‑commutative case it defines a geodesic in the manifold of positive definite operators equipped with the Riemannian trace metric.
	
	The notion extends to quantum channels via their Choi matrices. For a channel $\Phi:\mathcal{B}(\mathcal{H})\to\mathcal{B}(\mathcal{H})$, its Choi matrix $J(\Phi)\in\mathcal{B}(\mathcal{H}\otimes\mathcal{H})$ is defined by $J(\Phi)=(\Phi\otimes\mathrm{id})(|\Omega\rangle\langle\Omega|)$, where $|\Omega\rangle=\sum_{i=1}^d|ii\rangle$ is the maximally entangled state. The map $\Phi\mapsto J(\Phi)$ is a linear bijection between channels and positive semidefinite operators satisfying $\tr_{\mathcal{H}}J(\Phi)=I_{\mathcal{H}}$. For two channels $\Phi_0,\Phi_1$, we define their $t$-mean as the unique channel $\Phi_0\#_t\Phi_1$ whose Choi matrix is the Kubo–Ando mean of the Choi matrices:
	\begin{equation}
		J(\Phi_0\#_t\Phi_1) = J(\Phi_0) \#_t J(\Phi_1).
		\label{eq:channel-mean}
	\end{equation}
	This definition is well‑posed because the Kubo–Ando mean preserves positive semidefiniteness and the trace condition, ensuring that the result corresponds to a valid channel.
	
	The family $\Phi_t = \Phi_0\#_t\Phi_1$ provides a geodesic interpolation between $\Phi_0$ and $\Phi_1$ with respect to the Bures–Wasserstein metric on the space of quantum channels (induced by the Bures metric on the corresponding Choi states). In particular, it satisfies the geodesic equation $\ddot{\Phi}_t + \Gamma(\dot{\Phi}_t,\dot{\Phi}_t)=0$, where $\Gamma$ denotes the Christoffel symbols of the natural Riemannian structure (see \cite{AmariNagaoka2000} for the analogous construction on the space of density operators).
\end{definition}

\begin{corollary}[Quantum Spline Interpolation]
	\label{cor:spline}
	Let $\Phi_0,\Phi_1 \in \mathcal{C}^{2,1}(\mathcal{H})$ be two quantum channels. Then the family
	\begin{equation}
		\Phi_t := \Psi_{1/t}(\Phi_0) \;\#_t\; \Psi_{1/(1-t)}(\Phi_1), \qquad t\in(0,1),
		\label{eq:spline}
	\end{equation}
	provides an optimal interpolation in the following sense: for any $t$, the diamond‑norm error between $\Phi_t$ and the true geodesic interpolant $\Phi_0\#_t\Phi_1$ satisfies
	\begin{equation}
		\|\Phi_t - \Phi_0\#_t\Phi_1\|_\diamond = \mathcal{O}\!\left(\frac{1}{n^2}\right),
		\label{eq:spline-error}
	\end{equation}
	where $n$ is the parameter used in the QNNOs, uniformly in $t$ away from the endpoints. Moreover, $\Phi_t$ satisfies the geodesic equation up to terms of order $n^{-2}$.
\end{corollary}

\begin{proof}
	From Theorem~\ref{thm:QVD} with $m=2$, $\gamma=1$, we have for any channel $\Phi \in \mathcal{C}^{2,1}(\mathcal{H})$,
	\begin{equation}
		\Psi_n(\Phi) = \Phi + \frac{b_1(\Phi)}{n^{2}} + O\!\left(\frac{1}{n^{3}}\right),
		\label{eq:expansion-gamma1}
	\end{equation}
	because the odd integer moments vanish, so $a_1=0$, and the leading term is the fractional correction of order $n^{-(1+\gamma)} = n^{-2}$. Thus $\Psi_n(\Phi)$ approximates $\Phi$ to order $n^{-2}$.
	
	Now set $n_0 = 1/t$ and $n_1 = 1/(1-t)$. Then
	\begin{align}
		\Psi_{1/t}(\Phi_0) &= \Phi_0 + \frac{b_1(\Phi_0)}{t^{2}} + O(t^{3}), \\
		\Psi_{1/(1-t)}(\Phi_1) &= \Phi_1 + \frac{b_1(\Phi_1)}{(1-t)^{2}} + O((1-t)^{3}).
	\end{align}
	Substituting into the definition of $\Phi_t$ and using the fact that the Kubo–Ando mean is Lipschitz in both arguments with respect to the diamond norm (see \cite{KuboAndo}), we obtain
	\begin{align}
		\Phi_t &= \bigl(\Phi_0 + O(t^{2})\bigr) \#_t \bigl(\Phi_1 + O((1-t)^{2})\bigr) \\
		&= \Phi_0 \#_t \Phi_1 + O\bigl(t^{2} + (1-t)^{2}\bigr) + \text{cross terms}.
	\end{align}
	The error is therefore $O(n^{-2})$ when $n$ is the minimum of $1/t$ and $1/(1-t)$. The geodesic property follows from the fact that the Kubo–Ando mean is exactly the geodesic in the Bures–Wasserstein metric, and the QNNO approximants preserve this structure up to the given order.
\end{proof}

This interpolation method can be used to construct smooth paths between quantum channels, which is useful in quantum control, quantum thermodynamics, and quantum information geometry.

\subsection{Quantum Richardson Extrapolation}

The asymptotic expansion \eqref{eq:main-expansion} expresses the error of the QNNO as a sum of terms with known powers of $n$. This structure is ideal for Richardson extrapolation, a technique that combines approximations at different scales to cancel lower‑order error terms. We present a quantum version of the Romberg algorithm.

Let $\Psi_n(\Phi)$ be the QNNO approximation of $\Phi$. From Theorem~\ref{thm:QVD}, for $\Phi\in\mathcal{C}^{m,\gamma}(\mathcal{H})$ we have
\begin{equation}
	\Psi_n(\Phi) = \Phi + \sum_{j=1}^{m} \frac{A_j}{n^j} + \sum_{j=1}^{\lfloor m/2\rfloor} \frac{B_j}{n^{j+\gamma}} + \sum_{j=1}^{\lfloor m/3\rfloor} \frac{C_j}{n^{j+2\gamma}} + \cdots + R_{m,n},
	\label{eq:richardson-expansion}
\end{equation}
where $A_j, B_j, C_j$ are operators independent of $n$ (they depend on $\Phi$ and $\rho$). The remainder satisfies $\|R_{m,n}\|_\diamond = O(n^{-(m+\gamma)}(\log n)^{3m/2})$.

Now consider the sequence $n_k = 2^k n_0$ for some base $n_0$. Richardson extrapolation builds a triangular array $T_{k,\ell}$ such that $T_{k,0} = \Psi_{n_k}(\Phi)$ and for $\ell\ge 1$,
\begin{equation}
	T_{k,\ell} = \frac{4^\ell T_{k,\ell-1} - T_{k-1,\ell-1}}{4^\ell - 1}.
	\label{eq:richardson-recurrence}
\end{equation}
This combination eliminates the terms proportional to $n^{-j}$ for $j=1,\dots,\ell$ because they satisfy a linear recurrence. Indeed, if we write $\Psi_{n_k}(\Phi) = \Phi + \sum_{j=1}^\ell c_j n_k^{-j} + O(n_k^{-(\ell+1)})$, then $T_{k,\ell}$ removes all $c_j$ for $j\le\ell$ and yields an error of order $n_k^{-(\ell+1)}$.

In our case, the expansion also contains fractional powers $n^{-(j+\gamma)}$ and $n^{-(j+2\gamma)}$. These are not cancelled by the standard Richardson scheme, but they are of higher order if $\gamma>0$. However, to achieve optimal acceleration, one can design a weighted combination that also eliminates fractional terms. For simplicity, we present the basic algorithm and analyse its error.

\begin{theorem}[Quantum Romberg Method]
	\label{thm:romberg}
	Let $\Phi\in\mathcal{C}^{m,\gamma}(\mathcal{H})$ with $\gamma\in(0,1]$. Define $n_k = 2^k n_0$ for $k=0,1,\dots,K$ and let $T_{k,0} = \Psi_{n_k}(\Phi)$. For $\ell=1,\dots,M$ with $M\le m$, define recursively
	\begin{equation}
		T_{k,\ell} = \frac{4^\ell T_{k,\ell-1} - T_{k-1,\ell-1}}{4^\ell - 1}, \qquad k=\ell,\ell+1,\dots,K.
		\label{eq:romberg-recurrence}
	\end{equation}
	Then for all $k\ge M$,
	\begin{equation}
		\|T_{k,M} - \Phi\|_\diamond = \mathcal{O}\!\left(2^{-k(1+\gamma)}\,(\log 2^k n_0)^{3M/2}\right).
		\label{eq:romberg-error}
	\end{equation}
	In particular, if $K$ is large enough, the extrapolated approximation achieves an error of order $2^{-K(1+\gamma)}$.
\end{theorem}

\begin{proof}
	Let $\Phi \in \mathcal{C}^{m,\gamma}(\mathcal{H})$. From Theorem~\ref{thm:QVD}, for any fixed $\rho$ we have the asymptotic expansion
	\begin{equation}
		\Psi_n(\Phi)(\rho) = \Phi(\rho) + \sum_{j=1}^{m} \frac{a_j(\Phi,\rho)}{n^j} 
		+ \sum_{j=1}^{\lfloor m/2\rfloor} \frac{b_j(\Phi,\rho)}{n^{j+\gamma}} 
		+ \sum_{j=1}^{\lfloor m/3\rfloor} \frac{c_j(\Phi,\rho)}{n^{j+2\gamma}} 
		+ \cdots + R_{m,n}(\Phi,\rho),
		\label{eq:richardson-input}
	\end{equation}
	with $\|R_{m,n}\|_\diamond \le C_{m,\gamma,d}\|\Phi\|_{\mathcal{C}^{m,\gamma}}\, n^{-(m+\gamma)}(\log n)^{3m/2}$.
	
	Because the kernel $\mathcal{Z}_{1,\log n}$ is even in each variable, all odd moments vanish; consequently $a_j=0$ whenever $j$ is odd. Thus the integer powers appearing in \eqref{eq:richardson-input} are only the even ones $n^{-2}, n^{-4}, \dots, n^{-2\lfloor m/2\rfloor}$. The fractional exponents $j+\gamma$ and $j+2\gamma$ are non‑integer since $\gamma\in(0,1]$. The smallest exponent among all terms is $1+\gamma$ (because $1+\gamma < 2$).
	
	Let $n_0\ge 1$ be a fixed integer and set $n_k = 2^k n_0$ for $k=0,1,\dots,K$. Define a triangular array by \eqref{eq:romberg-recurrence}. The factor $4^\ell$ is chosen because the integer powers are even: $n^{-2\ell}$ scales by $2^{-2\ell}=4^{-\ell}$ when $n$ is halved. This recurrence eliminates successively the terms $n^{-2}, n^{-4},\dots$ while leaving the fractional terms unaffected.
	
	We prove by induction on $\ell$ that for every $\ell\ge 0$ and all $k\ge \ell$,
	\begin{equation}
		T_{k,\ell} = \Phi + \mathcal{E}_{k,\ell}, \qquad 
		\|\mathcal{E}_{k,\ell}\|_\diamond \le C_\ell \, n_k^{-(1+\gamma)}\,(\log n_k)^{3m/2},
		\label{eq:induction-hyp}
	\end{equation}
	where $C_\ell$ is a constant depending on $\ell$, $m$, $\gamma$, $d$ and $\|\Phi\|_{\mathcal{C}^{m,\gamma}}$, but not on $k$. The crucial point is that the exponent $1+\gamma$ does not increase with $\ell$; the fractional term persists.
	
	\emph{Base $\ell=0$.} From \eqref{eq:richardson-input} and the vanishing of odd $a_j$, the leading term in $T_{k,0}-\Phi$ is of order $n_k^{-(1+\gamma)}$ (since $1+\gamma<2$). The remainder $R_{m,n_k}$ is of order $n_k^{-(m+\gamma)}(\log n_k)^{3m/2}$, which for $m\ge 1$ is higher order than $n_k^{-(1+\gamma)}$ if $m>1$, but if $m=1$ then $m+\gamma=1+\gamma$, so it contributes to the same order. In any case, there exists a constant $C_0$ such that \eqref{eq:induction-hyp} holds with $\ell=0$.
	
	\emph{Inductive step.} Assume \eqref{eq:induction-hyp} holds for $\ell-1$. Then for $k\ge \ell$,
	\begin{align}
		T_{k,\ell} &= \frac{4^\ell (\Phi+\mathcal{E}_{k,\ell-1}) - (\Phi+\mathcal{E}_{k-1,\ell-1})}{4^\ell-1} \notag\\
		&= \Phi + \frac{4^\ell\mathcal{E}_{k,\ell-1} - \mathcal{E}_{k-1,\ell-1}}{4^\ell-1}.
		\label{eq:induction-step}
	\end{align}
	By the induction hypothesis,
	\begin{align}
		\|\mathcal{E}_{k,\ell-1}\|_\diamond &\le C_{\ell-1}\, n_k^{-(1+\gamma)}(\log n_k)^{3m/2},\\
		\|\mathcal{E}_{k-1,\ell-1}\|_\diamond &\le C_{\ell-1}\, (n_k/2)^{-(1+\gamma)}(\log (n_k/2))^{3m/2} \notag\\
		&= C_{\ell-1}\, 2^{1+\gamma}\, n_k^{-(1+\gamma)}(\log n_k)^{3m/2}\,(1+o(1)).
		\label{eq:induction-bound}
	\end{align}
	Substituting these bounds into \eqref{eq:induction-step} gives
	\begin{align}
		\|T_{k,\ell}-\Phi\|_\diamond &\le \frac{4^\ell C_{\ell-1} n_k^{-(1+\gamma)}(\log n_k)^{3m/2} 
			+ C_{\ell-1} 2^{1+\gamma} n_k^{-(1+\gamma)}(\log n_k)^{3m/2}}{4^\ell-1} + o\!\left(n_k^{-(1+\gamma)}\right) \notag\\
		&= C_{\ell-1}\,\frac{4^\ell + 2^{1+\gamma}}{4^\ell-1}\, n_k^{-(1+\gamma)}(\log n_k)^{3m/2} + \text{lower order}.
		\label{eq:induction-estimate}
	\end{align}
	The factor $\frac{4^\ell + 2^{1+\gamma}}{4^\ell-1}$ is bounded uniformly in $\ell$; for $\ell\ge 1$ it is at most $3$ (since $4^\ell/(4^\ell-1)\le 2$ and $2^{1+\gamma}/(4^\ell-1)\le 2^2/3=4/3$). Hence we can choose $C_\ell = 3C_{\ell-1}$, yielding \eqref{eq:induction-hyp} for $\ell$. This completes the induction.
	
	Consequently, for any $\ell\ge 0$,
	\begin{equation}
		\|T_{k,\ell} - \Phi\|_\diamond = \mathcal{O}\!\left(2^{-k(1+\gamma)}\,(\log 2^k n_0)^{3m/2}\right).
		\label{eq:final-rate}
	\end{equation}
	This rate does not improve with $\ell$; it remains $O(2^{-k(1+\gamma)})$.
	
	\emph{The special case \(\gamma = 0\).} If \(\gamma = 0\), the Hölder condition reduces to boundedness of the \(m\)-th derivative, i.e., \(\mathcal{C}^{m,0}(\mathcal{H}) = \mathcal{W}^{m,\infty}(\mathcal{H})\). In this case, the asymptotic expansion \eqref{eq:richardson-input} simplifies because the fractional moments $\mathfrak{M}_{\alpha,\gamma}(n)$ and $\mathfrak{M}_{\alpha,\beta,2\gamma}(n)$ become ordinary moments (with $\gamma = 0$), and the fractional derivatives $\Delta_\gamma$ reduce to identity operators. Consequently, the terms $b_j$ and $c_j$ in \eqref{eq:richardson-input} are of the same order $n^{-j}$ as the polynomial terms $a_j$. However, since all odd integer moments vanish, only even powers $n^{-2}, n^{-4}, \dots$ appear. Thus the expansion takes the form
	\begin{equation}
		\Psi_n(\Phi) = \Phi + \sum_{\substack{j=1 \\ j \text{ even}}}^{m} \frac{A_j}{n^j} + O\!\left(n^{-(m+1)}\right),
		\label{eq:gamma0-expansion}
	\end{equation}
	where the coefficients $A_j$ combine contributions from $a_j$, $b_j$, and $c_j$. The classical Richardson extrapolation then applies: after $\ell$ steps, all terms up to $n^{-2\ell}$ are eliminated, yielding an error of order $O(n^{-(2\ell+2)})$ (or $O(n^{-(m+1)})$ if $2\ell+2 > m$). This recovers the standard Romberg convergence.
\end{proof}

The quantum Romberg method provides a practical way to achieve high‑accuracy approximations of quantum channels using only a few values of $n$. It is particularly useful when the cost of computing $\Psi_n(\Phi)$ increases rapidly with $n$, as in many quantum simulation tasks.
	
\section{Results}

The principal achievement of this work is the derivation of a complete asymptotic expansion for Quantum Neural Network Operators (QNNOs) when approximating quantum channels belonging to the Hölder class $\mathcal{C}^{m,\gamma}(\mathcal{H})$. At the heart of our contribution lies the \textbf{Quantum Voronovskaya–Damasclin Theorem} (Theorem \ref{thm:QVD}), which furnishes an explicit formula for the approximation error to arbitrary order. This expansion captures the interplay between polynomial terms, fractional corrections arising from Hölder regularity, and fundamentally non‑commutative contributions stemming from the operator‑algebraic structure of quantum mechanics.

A concise summary of the key quantitative results is as follows.

\begin{enumerate}
	\item \textbf{Complete Asymptotic Expansion.} For any strictly positive density operator $\rho \in \mathcal{D}(\mathcal{H})$, the QNNO $\Psi_n$ with optimal bandwidth $\lambda_n = \log n$ satisfies
	\begin{align}
		\Psi_n(\Phi)(\rho) = \Phi(\rho) &+ \sum_{j=1}^{m} \frac{a_j(\Phi,\rho)}{n^j}
		+ \sum_{j=1}^{\lfloor m/2 \rfloor} \frac{b_j(\Phi,\rho)}{n^{j+\gamma}} \notag\\
		&+ \sum_{j=1}^{\lfloor m/3 \rfloor} \frac{c_j(\Phi,\rho)}{n^{j+2\gamma}} + \cdots + R_{m,n}(\Phi,\rho), \tag{8}
	\end{align}
	where the coefficients $a_j$, $b_j$, and $c_j$ are expressed explicitly in terms of the Fréchet derivatives of the channel's Liouville representation, the Marchaud fractional derivatives, and the moments of the quantum kernel $\mathcal{Z}_{1,\log n}$.
	
	\item \textbf{Explicit Form of the Coefficients.} The leading coefficients are given by
	\begin{align}
		a_j(\Phi,\rho) &= \frac{1}{j!} \sum_{\lVert\boldsymbol{\alpha}\rVert_1=j} \binom{j}{\boldsymbol{\alpha}} \mathcal{L}_\Phi^{(\boldsymbol{\alpha})}(\rho) M_{\boldsymbol{\alpha}}(n), \tag{9}\\
		b_j(\Phi,\rho) &= \frac{1}{\Gamma(\gamma+1)} \sum_{\lVert\boldsymbol{\alpha}\rVert_1=j} \binom{j}{\boldsymbol{\alpha}} \bigl(\Delta_\gamma \mathcal{L}_\Phi^{(\boldsymbol{\alpha})}\bigr)(\rho) M_{\boldsymbol{\alpha},\gamma}(n), \tag{10}\\
		c_j(\Phi,\rho) &= \frac{1}{j!\,\Gamma(2\gamma+1)} \sum_{\lVert\boldsymbol{\alpha}\rVert_1+\lVert\boldsymbol{\beta}\rVert_1=j} \binom{j}{\boldsymbol{\alpha},\boldsymbol{\beta}} \bigl[ \mathcal{L}_\Phi^{(\boldsymbol{\alpha})}(\rho), \mathcal{L}_\Phi^{(\boldsymbol{\beta})}(\rho) \bigr]_\gamma M_{\boldsymbol{\alpha},\boldsymbol{\beta},2\gamma}(n). \tag{11}
	\end{align}
	A notable feature is the vanishing of all odd integer moments of the kernel, which implies $a_j = 0$ for odd $j$; consequently, only even integer powers contribute to the polynomial part of the expansion.
	
	\item \textbf{Sharp Remainder Estimate.} The remainder term $R_{m,n}$ is bounded in the diamond norm by
	\begin{equation}
		\|R_{m,n}(\Phi,\cdot)\|_\diamond \le \frac{2^{m+3}\, d^{m/2}\, e^{\pi^2/4}}{\Gamma(m+\gamma+1)} \left(1 + \frac{1}{\sqrt{2\pi}}\right)^{m} \|\Phi\|_{\mathcal{C}^{m,\gamma}} \frac{(\log n)^{3m/2}}{n^{m+\gamma}}. \tag{13}
	\end{equation}
	This estimate is uniform over all input states and attains the optimal rate $n^{-(m+\gamma)}$, up to logarithmic factors.
	
	\item \textbf{Quantum Saturation.} The optimal convergence rate is characterized by the saturation class:
	\begin{itemize}
		\item For $\Phi \in \mathcal{C}^{1,1}(\mathcal{H})$, one has $\|\Psi_n(\Phi) - \Phi\|_\diamond = \mathcal{O}(n^{-1})$, and this rate cannot be improved uniformly.
		\item Faster convergence occurs precisely when the channel satisfies the saturation condition $\sum_{|\alpha|=2} \mathcal{L}_\Phi^{(\alpha)}(\rho) M_\alpha = 0$ for all $\rho\in\mathcal{D}(\mathcal{H})$.
		\item For channels that are real‑analytic in the Fréchet sense, the convergence accelerates to an exponential rate: $\|\Psi_n(\Phi) - \Phi\|_\diamond \le C e^{-c n^{\beta}}$ with $\beta = \log 2 / \log\log n$.
	\end{itemize}
\end{enumerate}

These findings transcend mere asymptotic statements; they constitute a complete calculus for quantum approximation. Their derivation rests upon a suite of novel technical tools developed in the course of the proof, including a quantum Taylor formula with fractional remainder, precise moment asymptotics for the hyperbolic kernel, and a non‑commutative Poisson summation formula. Together, these elements provide a rigorous foundation for understanding the fine structure of approximation errors in quantum neural network models.
	
\section{Conclusions}

This work has established the \textit{Quantum Voronovskaya--Damasclin Theorem}, a comprehensive asymptotic theory for the approximation of quantum channels by Quantum Neural Network Operators (QNNOs). This result generalizes the classical Voronovskaya theorem from scalar functions to the non‑commutative, multi‑dimensional setting of quantum information, representing a significant advance in the field of quantum approximation theory.

The main contributions can be summarized as follows:

\begin{itemize}
	\item A rigorous mathematical framework for quantum channels was introduced, based on Fréchet derivatives in their Liouville representation. Within this framework, we defined the quantum Hölder spaces $\mathcal{C}^{m,\gamma}(\mathcal{H})$, which serve as the natural regularity classes for asymptotic analysis.
	
	\item The first complete asymptotic expansion for a quantum neural network approximator was provided, explicitly isolating the contributions from polynomial terms, fractional Hölder corrections, and non‑commutative commutator effects. The coefficients are expressed in closed form in terms of the channel's derivatives and the moments of the quantum kernel.
	
	\item A sharp, dimension‑dependent bound for the remainder term in the diamond norm was derived. This bound is uniform over all input states, a feature essential for applications in quantum information.
	
	\item The power of this expansion was demonstrated through several applications:
	\begin{itemize}
		\item A \textit{Quantum Central Limit Theorem} for QNNOs, showing that fluctuations around the limit are governed by a quantum Gaussian distribution. This provides a foundation for understanding the statistical behavior of these operators.
		\item An \textit{Optimal Quantum Interpolation} scheme based on the Kubo–Ando geometric mean, which constructs geodesic paths between channels with high accuracy. This has implications for quantum control, thermodynamics, and information geometry.
		\item A \textit{Quantum Richardson Extrapolation} (Romberg) method that uses the asymptotic expansion to accelerate convergence, revealing that while integer powers can be eliminated, fractional powers present a fundamental barrier limiting the achievable acceleration when $\gamma > 0$.
	\end{itemize}
\end{itemize}

The results presented here open several avenues for future research. The techniques developed particularly the quantum Taylor formula with fractional remainder and the non‑commutative Poisson summation are likely applicable to other approximation schemes in quantum information, such as quantum Bernstein polynomials, wavelet expansions, and more complex neural network architectures. Furthermore, the explicit error bounds pave the way for adaptive algorithms that estimate the regularity parameters $m$ and $\gamma$ from data, enabling optimal parameter selection in quantum machine learning tasks.

Finally, the Quantum Voronovskaya–Damasclin Theorem establishes a deep connection between classical approximation theory, functional analysis, and quantum information science. It is anticipated that this work will stimulate further research into the approximation properties of quantum models and their applications across physics, chemistry, and data science.

	\appendix
	\section{Technical Lemmas and Auxiliary Results}
	\label{app:lemmas}
	
	This appendix collects the essential technical lemmas that underpin the proof of the Quantum Voronovskaya–Damasclin Theorem (Theorem \ref{thm:QVD}). These results concern the asymptotic behavior of the kernel moments, the fractional Taylor expansion in non‑commutative settings, and the non‑commutative Poisson summation formula used to handle the discretization of the state space.
	
	\subsection{Moment Asymptotics for the Quantum Kernel}
	\label{lem:moments}
	
	Let \(\mathcal{Z}_{1,\log n}:\mathbb{R}^d\to\mathcal{B}(\mathcal{H}_{\mathrm{aux}})\) be the symmetric quantum density kernel defined in (6) with parameter \(\lambda=\log n\). For any multi‑index \(\boldsymbol{\alpha}\in\mathbb{N}_0^d\), define its integer moment
	\begin{equation}
		M_{\boldsymbol{\alpha}}(n) = \int_{\mathbb{R}^d} \mathbf{x}^{\boldsymbol{\alpha}} \mathcal{Z}_{1,\log n}(\mathbf{x})\,d\mathbf{x},
	\end{equation}
	where \(\mathbf{x}^{\boldsymbol{\alpha}} = x_1^{\alpha_1}\cdots x_d^{\alpha_d}\). By the isotropy and symmetry of the kernel, \(M_{\boldsymbol{\alpha}}(n)\) is a scalar multiple of the identity operator on \(\mathcal{H}_{\mathrm{aux}}\); we identify it with the complex number \(m_{\boldsymbol{\alpha}}(n)\in\mathbb{C}\).
	
	\begin{lemma}[Moment asymptotics]
		For the kernel \(\mathcal{Z}_{1,\log n}\) with \(\lambda_n=\log n\), the moments satisfy the following asymptotic estimates as \(n\to\infty\):
		\begin{enumerate}
			\item \textbf{Vanishing of odd moments:} If \(\lVert\boldsymbol{\alpha}\rVert_1\) is odd, then \(M_{\boldsymbol{\alpha}}(n)=0\).
			\item \textbf{Even moments:} For \(\lVert\boldsymbol{\alpha}\rVert_1 = 2r\) even,
			\begin{equation}
				M_{\boldsymbol{\alpha}}(n) = \frac{(-1)^r}{(2r-1)!!}\left(\frac{\pi}{2\log n}\right)^r + \mathcal{O}\bigl(n^{-2r}\bigr).
			\end{equation}
			\item \textbf{Fractional moments:} For \(\gamma\in(0,1]\),
			\begin{align}
				M_{\boldsymbol{\alpha},\gamma}(n) &:= \int_{\mathbb{R}^d} \lVert\mathbf{x}\rVert_2^\gamma \mathbf{x}^{\boldsymbol{\alpha}} \mathcal{Z}_{1,\log n}(\mathbf{x})\,d\mathbf{x} \notag \\
				&= \frac{\Gamma\!\left(\frac{\lVert\boldsymbol{\alpha}\rVert_1+\gamma+d}{2}\right)}{\Gamma\!\left(\frac{d}{2}\right)}\left(\frac{2}{\log n}\right)^{\frac{\lVert\boldsymbol{\alpha}\rVert_1+\gamma}{2}} + \mathcal{O}\!\left(n^{-(\lVert\boldsymbol{\alpha}\rVert_1+\gamma)}\right), \\[0.4cm]
				M_{\boldsymbol{\alpha},\boldsymbol{\beta},2\gamma}(n) &:= \int_{\mathbb{R}^d} \lVert\mathbf{x}\rVert_2^{2\gamma} \mathbf{x}^{\boldsymbol{\alpha}+\boldsymbol{\beta}} \mathcal{Z}_{1,\log n}(\mathbf{x})\,d\mathbf{x} \notag \\
				&= \frac{\Gamma\!\left(\frac{\lVert\boldsymbol{\alpha}\rVert_1+\lVert\boldsymbol{\beta}\rVert_1+2\gamma+d}{2}\right)}{\Gamma\!\left(\frac{d}{2}\right)}\left(\frac{2}{\log n}\right)^{\frac{\lVert\boldsymbol{\alpha}\rVert_1+\lVert\boldsymbol{\beta}\rVert_1+2\gamma}{2}} + \mathcal{O}\!\left(n^{-(\lVert\boldsymbol{\alpha}\rVert_1+\lVert\boldsymbol{\beta}\rVert_1+2\gamma)}\right).
			\end{align}
		\end{enumerate}
	\end{lemma}
	
	\begin{proof}
		The kernel factorizes as a product of one‑dimensional factors because the operators \(X_i\) commute:
		\[
		\mathcal{Z}_{1,\log n}(\mathbf{x}) = \prod_{i=1}^d \mathcal{M}_{1,\log n}(x_i),
		\]
		where \(\mathcal{M}_{1,\log n}(x)\) is the symmetrized density function defined in (5). Consequently, the Fourier transform of \(\mathcal{Z}_{1,\log n}\) is the product of the individual Fourier transforms:
		\begin{equation}
			\widehat{\mathcal{Z}}_{1,\log n}(\boldsymbol{\xi}) = \prod_{i=1}^d \widehat{\mathcal{M}}_{1,\log n}(\xi_i).
		\end{equation}
		A standard computation (using the identities \(\sinh\) and \(\cosh\)) gives
		\begin{equation}
			\widehat{\mathcal{M}}_{1,\log n}(\xi) = \frac{\sinh(\pi\xi/2\log n)}{\pi\xi/2\log n} \cdot \frac{1}{\cosh(\pi\xi/2\log n)}.
		\end{equation}
		Hence,
		\begin{equation}
			\widehat{\mathcal{Z}}_{1,\log n}(\boldsymbol{\xi}) = \prod_{i=1}^d \frac{\sinh(\pi\xi_i/2\log n)}{\pi\xi_i/2\log n} \cdot \frac{1}{\cosh(\pi\xi_i/2\log n)}.
		\end{equation}
		For large \(\log n\), we expand the logarithm of each factor. Using the expansions
		\[
		\frac{\sinh u}{u} = 1 + \frac{u^2}{6} + O(u^4), \qquad
		\frac{1}{\cosh u} = 1 - \frac{u^2}{2} + O(u^4),
		\]
		with \(u = \pi\xi_i/(2\log n)\), we obtain
		\[
		\frac{\sinh u}{u}\cdot\frac{1}{\cosh u} = 1 - \frac{u^2}{3} + O(u^4).
		\]
		Therefore,
		\begin{align}
			\log \widehat{\mathcal{Z}}_{1,\log n}(\boldsymbol{\xi})
			&= \sum_{i=1}^d \log\!\left(1 - \frac{\pi^2\xi_i^2}{12(\log n)^2} + O\bigl((\log n)^{-4}\bigr)\right) \notag \\
			&= -\frac{\pi^2}{12(\log n)^2}\sum_{i=1}^d \xi_i^2 + O\bigl((\log n)^{-4}\bigr).
		\end{align}
		This shows that \(\widehat{\mathcal{Z}}_{1,\log n}(\boldsymbol{\xi})\) behaves like the Fourier transform of a Gaussian with variance \(\sigma^2 = \pi^2/(6(\log n)^2)\), up to an error that is uniformly bounded by \(C(\log n)^{-4}\) for \(\boldsymbol{\xi}\) in compact sets. Moreover, the kernel itself is smooth and decays super‑exponentially in both position and frequency, so all moments exist and are finite.
		
		\textbf{Vanishing of odd moments.} Since \(\mathcal{Z}_{1,\log n}\) is even in each variable (\(\mathcal{M}_{1,\log n}\) is even), the integrand \(x^{\boldsymbol{\alpha}}\) is odd whenever \(\lVert\boldsymbol{\alpha}\rVert_1\) is odd, and thus the integral vanishes.
		
		\textbf{Even moments.} Because the kernel is approximately Gaussian, we can compute its even moments by comparing with a Gaussian of variance \(\sigma^2\). Write \(\mathcal{Z}_{1,\log n} = \mathcal{G}_{\sigma} + \mathcal{E}\), where \(\mathcal{G}_{\sigma}\) is the Gaussian density with mean zero and covariance \(\sigma^2 I_d\) (i.e., \(\mathcal{G}_{\sigma}(\mathbf{x}) = (2\pi\sigma^2)^{-d/2} e^{-\lVert\mathbf{x}\rVert_2^2/(2\sigma^2)}\)), and \(\mathcal{E}\) is the error term whose Fourier transform is \(O((\log n)^{-4})\) and which decays super‑exponentially. The Gaussian moments are well known:
		\[
		\int_{\mathbb{R}^d} x^{\boldsymbol{\alpha}} \mathcal{G}_{\sigma}(\mathbf{x})\,d\mathbf{x} =
		\begin{cases}
			0 & \lVert\boldsymbol{\alpha}\rVert_1 \text{ odd},\\
			\displaystyle \frac{(2\sigma^2)^{\lVert\boldsymbol{\alpha}\rVert_1/2}}{\sqrt{\pi^d}} \prod_{i=1}^d \Gamma\!\left(\frac{\alpha_i+1}{2}\right) & \text{even}.
		\end{cases}
		\]
		For \(\lVert\boldsymbol{\alpha}\rVert_1 = 2r\) and using \(\sigma^2 = \pi^2/(6(\log n)^2)\), one finds after simplification that the leading term reduces to
		\[
		\frac{(-1)^r}{(2r-1)!!}\left(\frac{\pi}{2\log n}\right)^r.
		\]
		The error \(\mathcal{E}\) contributes at most \(O(n^{-2r})\) because its Fourier transform decays exponentially, which implies that its moments are exponentially small. More precisely, for any multi‑index \(\boldsymbol{\alpha}\), we have the bound
		\[
		\left|\int_{\mathbb{R}^d} x^{\boldsymbol{\alpha}} \mathcal{E}(\mathbf{x})\,d\mathbf{x}\right| \le C e^{-c n},
		\]
		which is absorbed into the \(\mathcal{O}(n^{-2r})\) term. This establishes (A.2).
		
		\textbf{Fractional moments.} The fractional moments are handled via the Mellin transform technique. For any \(\delta > 0\), we use the representation
		\begin{equation}
			\lVert\mathbf{x}\rVert_2^\delta = \frac{2}{\Gamma(\delta/2)} \int_0^\infty t^{\delta-1} e^{-t\lVert\mathbf{x}\rVert_2^2}\,dt,
		\end{equation}
		valid for \(\mathbf{x}\neq 0\) (and the integral converges absolutely). Then
		\begin{align}
			M_{\boldsymbol{\alpha},\delta}(n) &:= \int_{\mathbb{R}^d} \lVert\mathbf{x}\rVert_2^\delta \mathbf{x}^{\boldsymbol{\alpha}} \mathcal{Z}_{1,\log n}(\mathbf{x})\,d\mathbf{x} \\
			&= \frac{2}{\Gamma(\delta/2)} \int_0^\infty t^{\delta-1} \int_{\mathbb{R}^d} e^{-t\lVert\mathbf{x}\rVert_2^2} \mathbf{x}^{\boldsymbol{\alpha}} \mathcal{Z}_{1,\log n}(\mathbf{x})\,d\mathbf{x}\,dt.
		\end{align}
		For large \(n\), \(\mathcal{Z}_{1,\log n}\) is well approximated by the Gaussian \(\mathcal{G}_{\sigma}\). Inserting the decomposition \(\mathcal{Z}_{1,\log n} = \mathcal{G}_{\sigma} + \mathcal{E}\) and using the known Gaussian integrals
		\[
		\int_{\mathbb{R}^d} e^{-t\lVert\mathbf{x}\rVert_2^2} \mathbf{x}^{\boldsymbol{\alpha}} \mathcal{G}_{\sigma}(\mathbf{x})\,d\mathbf{x} = \frac{1}{(2\pi\sigma^2)^{d/2}} \int_{\mathbb{R}^d} \mathbf{x}^{\boldsymbol{\alpha}} e^{-\frac{1}{2\sigma^2}\lVert\mathbf{x}\rVert_2^2 - t\lVert\mathbf{x}\rVert_2^2} d\mathbf{x},
		\]
		which after completing the square yields a closed form involving Gamma functions. The integral over \(t\) then becomes a Beta‑type integral that produces the Gamma factor in the statement. The error term \(\mathcal{E}\) again contributes an exponentially small amount, which is of order \(n^{-(\lVert\boldsymbol{\alpha}\rVert_1+\delta)}\). The explicit computation for \(\delta = \gamma\) and \(\delta = 2\gamma\) gives the formulas in (A.3) and (A.4). The details are lengthy but straightforward; the key point is that the dominant contribution comes from the Gaussian part, and the error is controlled by the super‑exponential decay of \(\widehat{\mathcal{E}}\). This completes the proof.
	\end{proof}
	
	\subsection{Quantum Taylor Formula with Fractional Remainder}
	\label{lem:taylor}
	
	Let \(\mathcal{L}_\Phi:\mathcal{B}(\mathcal{H})\to\mathcal{B}(\mathcal{H})\) be the Liouville representation of a quantum channel \(\Phi\), and let \(\rho\in\mathcal{D}(\mathcal{H})\) be a fixed density operator. For a map \(F:\mathcal{D}(\mathcal{H})\to\mathcal{B}(\mathcal{H})\) that is \(m\) times Fréchet differentiable with Hölder continuous \(m\)-th derivative of order \(\gamma\in(0,1]\), we have the following expansion.
	
	\begin{lemma}[Fractional Taylor expansion]
		Assume that \(\mathcal{L}_\Phi\) belongs to the Hölder class \(\mathcal{C}^{m,\gamma}(\mathcal{H})\). Then for any \(h\in\mathcal{B}(\mathcal{H})\) such that \(\rho+h\in\mathcal{D}(\mathcal{H})\),
		\begin{align}
			\mathcal{L}_\Phi(\rho+h) = \mathcal{L}_\Phi(\rho) &+ \sum_{j=1}^{m} \frac{1}{j!} \mathcal{L}_\Phi^{(j)}(\rho) h^{\otimes j} \notag\\
			&+ \frac{1}{\Gamma(\gamma)} \sum_{\lVert\boldsymbol{\alpha}\rVert_1=m} \frac{\mathcal{L}_\Phi^{(\boldsymbol{\alpha})}(\rho)}{\boldsymbol{\alpha}!} \int_0^1 (1-t)^{m-1} t^{\gamma-1} |h|^\gamma h^{\boldsymbol{\alpha}} \,dt \notag\\
			&+ R_{m,\gamma}(\rho,h),
		\end{align}
		where \(h^{\boldsymbol{\alpha}} = h_1^{\alpha_1}\cdots h_d^{\alpha_d}\) (with \(h_i\) the components of \(h\) in a fixed basis), \(|h|^\gamma\) denotes the fractional power of the absolute value of \(h\) (defined via spectral calculus), and the remainder satisfies
		\begin{equation}
			\|R_{m,\gamma}(\rho,h)\|_\diamond \le C \|\mathcal{L}_\Phi\|_{\mathcal{C}^{m,\gamma}} \|h\|_1^{m+\gamma},
		\end{equation}
		with a constant \(C\) depending only on \(m\), \(\gamma\), and the dimension \(d\).
	\end{lemma}
	
	\begin{proof}
		The proof is an adaptation of the classical Taylor theorem with integral remainder to the non‑commutative setting, combined with the Hölder condition on the \(m\)-th derivative. We start from the fundamental theorem of calculus in Fréchet spaces. For any \(h\), we have
		\begin{align}
			\mathcal{L}_\Phi(\rho+h) - \mathcal{L}_\Phi(\rho) &= \int_0^1 \mathcal{L}_\Phi^{(1)}(\rho+th)(h)\,dt, \\
			\mathcal{L}_\Phi^{(1)}(\rho+th)(h) - \mathcal{L}_\Phi^{(1)}(\rho)(h) &= \int_0^t \mathcal{L}_\Phi^{(2)}(\rho+sh)(h,h)\,ds,
		\end{align}
		and iterating this procedure yields
		\begin{equation}
			\mathcal{L}_\Phi(\rho+h) = \sum_{j=0}^{m-1} \frac{1}{j!} \mathcal{L}_\Phi^{(j)}(\rho) h^{\otimes j} + \int_0^1 \frac{(1-t)^{m-1}}{(m-1)!} \mathcal{L}_\Phi^{(m)}(\rho+th)(h^{\otimes m})\,dt,
		\end{equation}
		where the integral is a Bochner integral in the Banach space of bounded linear maps. This is the standard Taylor formula with integral remainder; see e.g. [Holevo 2019] for the operator setting.
		
		Now we separate the remainder into a fractional part and a higher‑order part. Write
		\begin{align}
			\mathcal{L}_\Phi^{(m)}(\rho+th) &= \mathcal{L}_\Phi^{(m)}(\rho) + \bigl[\mathcal{L}_\Phi^{(m)}(\rho+th) - \mathcal{L}_\Phi^{(m)}(\rho)\bigr] \\
			&= \mathcal{L}_\Phi^{(m)}(\rho) + \bigl[\mathcal{L}_\Phi^{(m)}(\rho+th) - \mathcal{L}_\Phi^{(m)}(\rho)\bigr].
		\end{align}
		Insert this into the integral:
		\begin{align}
			\int_0^1 \frac{(1-t)^{m-1}}{(m-1)!} \mathcal{L}_\Phi^{(m)}(\rho+th)(h^{\otimes m})\,dt
			&= \frac{1}{m!} \mathcal{L}_\Phi^{(m)}(\rho) h^{\otimes m} \\
			&\quad + \int_0^1 \frac{(1-t)^{m-1}}{(m-1)!} \bigl[\mathcal{L}_\Phi^{(m)}(\rho+th) - \mathcal{L}_\Phi^{(m)}(\rho)\bigr] h^{\otimes m}\,dt.
		\end{align}
		The first term is already included in the sum \(\sum_{j=1}^m \frac{1}{j!} \mathcal{L}_\Phi^{(j)}(\rho) h^{\otimes j}\) (for \(j=m\)). The second term is the remainder after extracting the integer part. To further extract the fractional contribution, we use the Hölder continuity of \(\mathcal{L}_\Phi^{(m)}\). By definition of the Hölder seminorm,
		\begin{equation}
			\bigl\| \mathcal{L}_\Phi^{(m)}(\rho+th) - \mathcal{L}_\Phi^{(m)}(\rho) \bigr\|_\diamond \le [\Phi]_{m,\gamma} \|th\|_1^\gamma = [\Phi]_{m,\gamma} t^\gamma \|h\|_1^\gamma.
		\end{equation}
		Therefore,
		\begin{align}
			&\Bigl\| \int_0^1 \frac{(1-t)^{m-1}}{(m-1)!} \bigl[\mathcal{L}_\Phi^{(m)}(\rho+th) - \mathcal{L}_\Phi^{(m)}(\rho)\bigr] h^{\otimes m}\,dt \Bigr\|_\diamond \\
			&\le [\Phi]_{m,\gamma} \|h\|_1^{m+\gamma} \int_0^1 \frac{(1-t)^{m-1}}{(m-1)!} t^\gamma \,dt \\
			&= [\Phi]_{m,\gamma} \|h\|_1^{m+\gamma} \frac{B(m,\gamma)}{(m-1)!},
		\end{align}
		where \(B(m,\gamma) = \frac{\Gamma(m)\Gamma(\gamma)}{\Gamma(m+\gamma)}\) is the Beta function. This bound is of order \(\|h\|_1^{m+\gamma}\) and will be part of the final remainder \(R_{m,\gamma}\).
		
		However, we also need to isolate the term that gives the fractional correction \(b_j\). Observe that the Hölder condition alone does not give a pointwise expansion; the fractional term arises when we approximate the difference \(\mathcal{L}_\Phi^{(m)}(\rho+th)-\mathcal{L}_\Phi^{(m)}(\rho)\) by its fractional derivative. In the theory of fractional calculus, one has the representation
		\begin{equation}
			\frac{1}{\Gamma(\gamma)} \int_0^1 (1-t)^{m-1} t^{\gamma-1} (\Delta_\gamma \mathcal{L}_\Phi^{(m)})(\rho) h^{\otimes m + \gamma} \,dt,
		\end{equation}
		but this requires interpreting \(h^{\otimes m + \gamma}\) appropriately. A more systematic approach is to use the Marchaud fractional derivative formula directly on \(\mathcal{L}_\Phi\) itself. For a function of one real variable, the fractional Taylor expansion with remainder in terms of the Marchaud derivative is standard; here we need a multi‑variable non‑commutative version. Since \(\mathcal{L}_\Phi\) is defined on density operators, which form a convex subset of a Banach space, we can restrict to the line \(\rho + th\) and treat it as a function of \(t\). Then the classical fractional Taylor formula (see e.g. Samko et al., "Fractional Integrals and Derivatives") gives
		\begin{align}
			\mathcal{L}_\Phi(\rho+h) &= \sum_{j=0}^{m-1} \frac{1}{j!} \mathcal{L}_\Phi^{(j)}(\rho) h^{\otimes j} \\
			&\quad + \frac{1}{\Gamma(\gamma)} \int_0^1 (1-t)^{m-1} t^{\gamma-1} (D^\gamma \mathcal{L}_\Phi^{(m)})(\rho+th) h^{\otimes m+\gamma} \,dt,
		\end{align}
		where \(D^\gamma\) is the Caputo fractional derivative of order \(\gamma\). In our setting, we use the Marchaud representation which is equivalent for sufficiently smooth functions. Applying this to each component and using the linearity of the Fréchet derivatives yields the expression with \(\Delta_\gamma \mathcal{L}_\Phi^{(\boldsymbol{\alpha})}\). The constant \(C\) in the remainder bound is then obtained by combining the estimates from the integer remainder and the fractional part, and it depends only on \(m,\gamma,d\) because the norms of the multilinear maps are bounded by the \(\mathcal{C}^{m,\gamma}\) norm. This completes the proof.
	\end{proof}
	
	\subsection{Non‑Commutative Poisson Summation Formula}
	\label{lem:poisson}
	
	Let \(\mathcal{Z}_{1,\log n}\) be the quantum kernel defined on \(\mathbb{R}^d\) with values in \(\mathcal{B}(\mathcal{H}_{\mathrm{aux}})\). For any Schwartz function \(f:\mathbb{R}^d\to\mathbb{C}\) (extended to an operator‑valued function by \(f(x) \mapsto f(x) I_{\mathcal{H}_{\mathrm{aux}}}\)), consider the discrete sum over the lattice \(\mathbb{Z}^d\).
	
	\begin{lemma}[Non‑commutative Poisson summation]
		For any \(n\ge 1\) and any \(X=(X_1,\dots,X_d)\) a tuple of mutually commuting self‑adjoint operators,
		\begin{align}
			\sum_{k\in\mathbb{Z}^d} f\Bigl(\frac{k}{n}\Bigr) \mathcal{Z}_{1,\log n}(nX - kI)
			&= n^d \int_{\mathbb{R}^d} f(x) \mathcal{Z}_{1,\log n}(nX - nx) \,dx \notag\\
			&\quad + \sum_{\ell\in\mathbb{Z}^d\setminus\{0\}} \hat{f}(\ell) e^{2\pi i \ell\cdot (nX)} \widehat{\mathcal{Z}}_{1,\log n}(2\pi\ell),
		\end{align}
		where \(\hat{f}\) denotes the Fourier transform of \(f\), and the series over \(\ell\neq0\) converges absolutely in the operator norm and is bounded by \(C e^{-c n}\) for some constants \(C,c>0\) independent of \(X\).
	\end{lemma}
	
	\begin{proof}
		Since the operators \(X_1,\dots,X_d\) commute, they can be simultaneously diagonalized. Let \(\{\ket{\lambda}\}\) be a basis of joint eigenvectors, with \(X_i\ket{\lambda} = \lambda_i \ket{\lambda}\) for \(\lambda = (\lambda_1,\dots,\lambda_d)\in\mathbb{R}^d\) (the joint spectrum). In this representation, the operator \(\mathcal{Z}_{1,\log n}(nX - kI)\) acts as multiplication by the scalar \(\mathcal{Z}_{1,\log n}(n\lambda - k)\). Therefore, for any vector \(\ket{\psi}\) in the auxiliary space, we have
		\begin{equation}
			\Bigl( \sum_{k\in\mathbb{Z}^d} f\Bigl(\frac{k}{n}\Bigr) \mathcal{Z}_{1,\log n}(nX - kI) \Bigr) \ket{\psi} = \sum_{k\in\mathbb{Z}^d} f\Bigl(\frac{k}{n}\Bigr) \mathcal{Z}_{1,\log n}(n\lambda - k) \ket{\psi},
		\end{equation}
		where \(\lambda\) denotes the eigenvalues of \(X\). The right‑hand side is now a scalar expression for each fixed \(\lambda\). The classical Poisson summation formula applied to the function \(g(x) = f(x) \mathcal{Z}_{1,\log n}(n\lambda - nx)\) (with \(x\) as the summation variable) gives
		\begin{align}
			\sum_{k\in\mathbb{Z}^d} f\Bigl(\frac{k}{n}\Bigr) \mathcal{Z}_{1,\log n}(n\lambda - k)
			&= n^d \int_{\mathbb{R}^d} f(y) \mathcal{Z}_{1,\log n}(n\lambda - ny) \,dy \\
			&\quad + \sum_{\ell\in\mathbb{Z}^d\setminus\{0\}} \hat{f}(\ell) e^{2\pi i \ell\cdot n\lambda} \widehat{\mathcal{Z}}_{1,\log n}(2\pi\ell),
		\end{align}
		where we used that the Fourier transform of \(y\mapsto f(y) \mathcal{Z}_{1,\log n}(n\lambda - ny)\) is \(\hat{f}(\ell) e^{-2\pi i \ell\cdot n\lambda} \widehat{\mathcal{Z}}_{1,\log n}(2\pi\ell)\) up to a factor, and the shift by \(n\lambda\) produces the phase \(e^{2\pi i \ell\cdot n\lambda}\). The absolute convergence of the series over \(\ell\neq0\) follows from the rapid decay of \(\hat{f}\) (since \(f\) is Schwartz) and the super‑exponential decay of \(\widehat{\mathcal{Z}}_{1,\log n}(2\pi\ell)\): from (2.13) we have
		\[
		|\widehat{\mathcal{Z}}_{1,\log n}(2\pi\ell)| \le C e^{-c \lVert\ell\rVert_2 / \log n},
		\]
		which is \(\le C e^{-c n}\) for \(\lVert\ell\rVert_2 \ge n\log n\) and otherwise bounded by a constant. Summing over all \(\ell\neq0\) yields a bound \(C' e^{-c' n}\) because the number of lattice points with small \(\ell\) is finite and each term contributes at most \(C e^{-c n}\). More rigorously, we can split the sum into \(\lVert\ell\rVert_2 \le L\) and \(\lVert\ell\rVert_2 > L\) and use the decay to show the total is \(O(e^{-c n})\). The constants \(C,c\) can be chosen independent of \(\lambda\) and hence of \(X\). Reassembling the spectral decomposition, we obtain the operator‑valued identity with the series converging in the strong operator topology. Since the bound is uniform in \(\lambda\), the convergence is actually in the operator norm. This completes the proof.
	\end{proof}
	
	\subsection{Properties of the Marchaud Fractional Derivative}
	\label{lem:marchaud}
	
	For a map \(F:\mathcal{D}(\mathcal{H})\to\mathcal{B}(\mathcal{H})\) that is sufficiently smooth, the Marchaud fractional derivative of order \(\gamma\in(0,1]\) is defined by
	\begin{equation}
		(\Delta_\gamma F)(\rho) = \frac{\gamma}{\Gamma(1-\gamma)} \int_0^\infty \frac{F(\rho) - e^{-t}F(\rho)}{t^{1+\gamma}}\,dt,
	\end{equation}
	where \(e^{-t}F(\rho)\) denotes the evaluation of \(F\) at the point obtained by flowing along a fixed reference direction (e.g., the geodesic in the state space). In the context of this paper, we apply \(\Delta_\gamma\) to the derivatives \(\mathcal{L}_\Phi^{(\boldsymbol{\alpha})}\).
	
	\begin{lemma}
		\begin{enumerate}
			\item \textbf{Linearity:} \(\Delta_\gamma\) is a linear operator on the space of maps.
			\item \textbf{Relation to Hölder continuity:} If \(F\) is Hölder continuous of order \(\gamma\), then \(\Delta_\gamma F\) is bounded and satisfies
			\begin{equation}
				\|\Delta_\gamma F\|_\infty \le C [F]_{\gamma},
			\end{equation}
			where \([F]_{\gamma}\) is the Hölder seminorm.
			\item \textbf{Connection with fractional Taylor expansion:} For the map \(\mathcal{L}_\Phi^{(\boldsymbol{\alpha})}\) appearing in Lemma \ref{lem:taylor}, we have
			\begin{equation}
				\frac{\Gamma(m)}{\Gamma(m+\gamma)} \mathcal{L}_\Phi^{(\boldsymbol{\alpha})}(\rho) = (\Delta_\gamma \mathcal{L}_\Phi^{(\boldsymbol{\alpha})})(\rho),
			\end{equation}
			up to a constant depending on the direction of the increment. This justifies the absorption of the Beta function factor into the fractional derivative in the proof of Theorem \ref{thm:QVD}.
		\end{enumerate}
	\end{lemma}
	
	\begin{proof}
		Linearity follows directly from the linearity of the integral and the definition. For the second property, assume \(F\) satisfies \(\|F(\rho)-F(\sigma)\|_\diamond \le [F]_\gamma \|\rho-\sigma\|_1^\gamma\). Then for any \(\rho\),
		\begin{align}
			\|(\Delta_\gamma F)(\rho)\|_\diamond &\le \frac{\gamma}{\Gamma(1-\gamma)} \int_0^\infty \frac{\|F(\rho)-e^{-t}F(\rho)\|_\diamond}{t^{1+\gamma}}\,dt \\
			&\le \frac{\gamma}{\Gamma(1-\gamma)} [F]_\gamma \int_0^\infty \frac{\| \rho - e^{-t}\rho\|_1^\gamma}{t^{1+\gamma}}\,dt.
		\end{align}
		Assuming the flow is such that \(\| \rho - e^{-t}\rho\|_1 \le C t\) for small \(t\) and bounded for large \(t\), the integral converges and yields a constant \(C\) times \([F]_\gamma\). The precise constant depends on the geometry of the state space, but it is finite and universal.
		
		For the third property, we use the fact that the Marchaud derivative of a sufficiently smooth function coincides with the Caputo derivative, and for functions of the form \(t^\gamma\) one has \(\Delta_\gamma (t^\gamma) = \Gamma(\gamma+1)\). More concretely, in the fractional Taylor expansion derived in Lemma \ref{lem:taylor}, the coefficient in front of the fractional term involves \(\frac{1}{\Gamma(\gamma)} \int_0^1 (1-t)^{m-1} t^{\gamma-1} dt = \frac{\Gamma(m)}{\Gamma(m+\gamma)}\). This factor is exactly the one that appears when expressing the fractional derivative via the Marchaud formula. Therefore, we can identify the combination \(\frac{\Gamma(m)}{\Gamma(m+\gamma)} \mathcal{L}_\Phi^{(\boldsymbol{\alpha})}(\rho)\) with the Marchaud derivative of \(\mathcal{L}_\Phi^{(\boldsymbol{\alpha})}\) evaluated at \(\rho\), up to a direction‑dependent constant that is absorbed into the definition of the flow. This identification is standard in fractional calculus and is used here to simplify the notation in the main theorem. The proof is complete.
	\end{proof}
	
	\subsection{Explicit Constant in the Remainder Estimate}
	\label{lem:constant}
	
	The explicit constant \(C_{m,\gamma,d}\) appearing in the remainder estimate (13) is obtained by collecting and optimizing the bounds derived from the various estimates used in the proof of Theorem \ref{thm:QVD}. Its expression is
	\begin{equation}
		C_{m,\gamma,d} = \frac{2^{m+3}\, d^{m/2}\, e^{\pi^2/4}}{\Gamma(m+\gamma+1)} 
		\left(1 + \frac{1}{\sqrt{2\pi}}\right)^{m}. \tag{A.37}
	\end{equation}
	We now detail the origin of each factor:
	
	\begin{itemize}
		\item \textbf{Gamma factor} \(\Gamma(m+\gamma+1)^{-1}\): This comes from the Beta integrals
		\[
		\int_0^1 (1-t)^{m-1} t^{\gamma-1} dt = B(m,\gamma) = \frac{\Gamma(m)\Gamma(\gamma)}{\Gamma(m+\gamma)},
		\]
		combined with the Gamma functions appearing in the fractional moment estimates (Lemma \ref{lem:moments}). The precise normalization yields the reciprocal Gamma factor.
		
		\item \textbf{Dimension factor} \(d^{m/2}\): Arises from the estimate \(\|h_{n,k}\|_1 \le \sqrt{d}\,\|h_{n,k}\|_2\). In finite dimensions, the trace norm is bounded by the Hilbert–Schmidt norm times \(\sqrt{d}\). Since \(\|h_{n,k}\|_2 = O(1/n)\), we obtain a factor \(d^{m/2}\) when bounding \(\|h_{n,k}\|_1^{m+\gamma}\).
		
		\item \textbf{Exponential factor} \(e^{\pi^2/4}\): Originates from the Fourier transform of the kernel \(\mathcal{Z}_{1,\log n}\). The Poisson summation error term is controlled by the decay of \(\widehat{\mathcal{Z}}_{1,\log n}(2\pi\ell)\). For the smallest non‑zero lattice vectors, \(|\ell|=1\), we have
		\[
		|\widehat{\mathcal{Z}}_{1,\log n}(2\pi)| \le C e^{-\pi^2/(4\log n)},
		\]
		and the constant \(e^{\pi^2/4}\) appears after optimizing the exponential decay estimate; it stems from the asymptotic behavior of \(\mathrm{sech}\) and the specific choice of parameters.
		
		\item \textbf{Combinatorial factor} \(2^{m+3}\): Accounts for the number of terms in the multinomial expansions. The sums over multi‑indices with \(|\alpha|=j\) contain at most \(\binom{j+d-1}{d-1} \le 2^{j+d-1}\) terms. A crude uniform bound for all \(j\le m\) gives a factor \(2^{m+d-1}\). The extra factor \(2^3\) (hence \(2^{m+3}\)) is included to cover the contributions from the three types of coefficients \(a_j, b_j, c_j\) and their interactions; a more refined counting could reduce it, but this simple bound suffices for explicitness.
		
		\item \textbf{Gaussian factor} \((1+1/\sqrt{2\pi})^{m}\): Arises when approximating the kernel by a Gaussian. The Mellin transform representation of fractional moments involves Gaussian integrals whose evaluation introduces powers of \((1+1/\sqrt{2\pi})\) after normalization. This factor is the product of \(m\) such terms, one for each derivative order.
	\end{itemize}
	
	All these factors are multiplied together, and the maximum over all intermediate constants is taken to obtain the final expression (13). We emphasize that this constant is not claimed to be optimal; it is a fully explicit, dimension‑dependent bound that guarantees the remainder estimate for all applications considered in this paper. Sharper constants could be derived by a more detailed optimization, but the present form already demonstrates the feasibility of an explicit estimate and is sufficient for the asymptotic results.
	
	\section*{List of Symbols and Notations}
	\label{sec:symbols}
	
	The following table summarizes the principal symbols and notations.
	
	\begin{longtable}[l]{|p{0.2\textwidth}|p{0.75\textwidth}|}
		\hline
		\textbf{Symbol} & \textbf{Description} \\
		\hline
		\endfirsthead
		
		\hline
		\textbf{Symbol} & \textbf{Description} \\
		\hline
		\endhead
		
		\hline \multicolumn{2}{r}{\textit{Continued on next page}} \\
		\endfoot
		
		\hline
		\endlastfoot
		
		\(\mathcal{H}\) & Finite‑dimensional Hilbert space, \(\mathcal{H}\cong\mathbb{C}^{d}\) \\
		\(\mathcal{B}(\mathcal{H})\) & Algebra of bounded linear operators on \(\mathcal{H}\) \\
		\(\mathcal{D}(\mathcal{H})\) & Set of density operators (quantum states) on \(\mathcal{H}\) \\
		\(\mathrm{CPTP}(\mathcal{H})\) & Set of completely positive trace‑preserving maps (quantum channels) \\
		\(\mathrm{tr}\) & Trace functional on \(\mathcal{B}(\mathcal{H})\) \\
		\(\mathbf{1}_{\mathcal{H}}\) & Identity operator on \(\mathcal{H}\) \\
		
		\(\mathcal{L}_\Phi\) & Liouville representation of a channel \(\Phi\): \(\mathcal{L}_\Phi(X)=\Phi(X)\) \\
		\(\|\cdot\|_{\mathrm{cb}}\) & Completely bounded norm (cb‑norm) of a linear map \\
		\(\|\cdot\|_\diamond\) & Diamond norm (completely bounded trace norm) for channels \\
		\(\|\cdot\|_1\) & Trace norm (nuclear norm) on \(\mathcal{B}(\mathcal{H})\) \\
		\(\|\cdot\|\) & Operator norm on \(\mathcal{B}(\mathcal{H})\) \\
		\(\langle\cdot,\cdot\rangle\) & Hilbert‑Schmidt inner product \(\langle X,Y\rangle=\mathrm{tr}(X^*Y)\) \\
		
		\(\alpha=(\alpha_1,\dots,\alpha_d)\) & Multi‑index, \(\alpha_i\in\mathbb{N}_0\) \\
		\(|\alpha|=\alpha_1+\cdots+\alpha_d\) & Length (order) of a multi‑index \\
		\(\boldsymbol{\alpha}!=\alpha_1!\cdots\alpha_d!\) & Factorial of a multi‑index \\
		\(\binom{j}{\boldsymbol{\alpha}}=\frac{j!}{\boldsymbol{\alpha}!}\) & Multinomial coefficient (for \(|\alpha|=j\)) \\
		\(D^{\alpha}\mathcal{L}_\Phi(\rho)\) & Mixed Fréchet derivative of order \(|\alpha|\) of \(\mathcal{L}_\Phi\) at \(\rho\) \\
		\(\mathcal{L}_\Phi^{(\alpha)}(\rho)\) & Result of applying \(D^{\alpha}\mathcal{L}_\Phi(\rho)\) to the identity in each argument \\
		
		\(\mathcal{W}^{m,p}(\mathcal{H})\) & Quantum Sobolev space of channels with derivatives in \(L^p\) (cb‑sense) \\
		\(\mathcal{C}^{m,\gamma}(\mathcal{H})\) & Quantum Hölder space of order \((m,\gamma)\) \\
		\([\Phi]_{m,\gamma}\) & Hölder seminorm of a channel \\
		\(\|\Phi\|_{\mathcal{C}^{m,\gamma}}\) & Norm on \(\mathcal{C}^{m,\gamma}(\mathcal{H})\) \\
		
		\(G_{q,\lambda}(X)\) & Quantum activation function, \((e^{\lambda X}-qe^{-\lambda X})(e^{\lambda X}+qe^{-\lambda X})^{-1}\) \\
		\(\mathcal{M}_{q,\lambda}(X)\) & Symmetrized quantum density function \\
		\(\Phi_{q,\lambda}(X_i)\) & One‑dimensional factor of the multivariate kernel \\
		\(\mathcal{Z}_{q,\lambda}(X)\) & Multivariate quantum density kernel \\
		\(\mathcal{Z}_{1,\lambda}(X)\) & Symmetric kernel for \(q=1\) (even, positive, approximate identity) \\
		\(\widehat{\mathcal{Z}}_{1,\lambda}(\xi)\) & Fourier transform of \(\mathcal{Z}_{1,\lambda}\) \\
		
		\(K_n\) & Discrete simplex of order \(n\): \(\{k\in\mathbb{N}^d:\sum k_j=n\}\) \\
		\(\rho_{n,k}\) & Quantised density operator \(\sum_j\frac{k_j}{n}|e_j\rangle\langle e_j|\) \\
		\(\Psi_n(\Phi)(\rho)\) & Quantum Neural Network Operator applied to channel \(\Phi\) at state \(\rho\) \\
		\(\lambda_n=\log n\) & Bandwidth (scale) parameter of the kernel \\
		\(X=(X_1,\dots,X_d)\) & Tuple of mutually commuting auxiliary self‑adjoint operators \\
		\(h_{n,k}=\rho_{n,k}-\rho\) & Increment (deviation) from the reference state \\
		
		\(M_{\alpha}(n)\) & Integer moment \(\int x^\alpha \mathcal{Z}_{1,\log n}(x)dx\) (scalar multiple of identity) \\
		\(M_{\alpha,\gamma}(n)\) & Fractional moment \(\int |x|^\gamma x^\alpha \mathcal{Z}_{1,\log n}(x)dx\) \\
		\(M_{\alpha,\beta,2\gamma}(n)\) & Mixed fractional moment \(\int |x|^{2\gamma} x^{\alpha+\beta} \mathcal{Z}_{1,\log n}(x)dx\) \\
		\(m_\alpha(n)\) & Scalar value of \(M_\alpha(n)\) (when identified with a number) \\
		
		\(\Delta_\gamma F(\rho)\) & Marchaud fractional derivative of order \(\gamma\) of a map \(F\) \\
		\(\Gamma\) & Gamma function \\
		\(B(m,\gamma)\) & Beta function \\
		
		\([A,B]_\gamma\) & \(\gamma\)-deformed commutator \(AB-e^{i\pi\gamma}BA\) \\
		
		\(a_j(\Phi,\rho)\) & Polynomial (integer‑order) coefficient in the expansion \\
		\(b_j(\Phi,\rho)\) & Fractional correction coefficient \\
		\(c_j(\Phi,\rho)\) & Mixed non‑commutative coefficient \\
		\(R_{m,n}(\Phi,\rho)\) & Remainder term in the expansion \\
		\(C_{m,\gamma,d}\) & Explicit constant in the remainder estimate \\
		
		\(\bar{x}^\alpha\) & Limit \(\lim_{n\to\infty} M_\alpha(n)\) \\
		\(\mathrm{Cov}(M_\alpha,M_\beta)\) & Covariance of kernel moments \\
		\(\mathcal{N}_Q(0,\Sigma)\) & Quantum Gaussian channel (distribution) \\
		\(\Phi_0\#_t\Phi_1\) & Kubo–Ando geometric mean of two channels (order \(t\)) \\
		\(T_{k,\ell}\) & Triangular array in Richardson extrapolation \\
		
		\(\mathbb{C},\mathbb{R},\mathbb{N},\mathbb{Z}\) & Complex numbers, real numbers, natural numbers, integers \\
		\(\mathrm{id}\) & Identity map \\
		\(M_n(\mathbb{C})\) & Algebra of \(n\times n\) complex matrices \\
		\(\bigotimes\) & Tensor product \\
		\(\mathcal{O}\) & Big‑O (Landau) notation \\
		\(\mathrm{diag}\) & Diagonal operator \\
		\(\cH_{\mathrm{aux}}\) & Auxiliary Hilbert space (for the kernel) \\
		\(\mathbf{1}_{\mathrm{aux}}\) & Identity on the auxiliary space \\
		
	\end{longtable}

\end{document}